\documentclass[11pt,letterpaper]{article}

\usepackage{multirow,url}
\usepackage{amsfonts}
\usepackage[breaklinks]{hyperref}
\usepackage{fancyhdr}
\usepackage[margin=1.0in]{geometry}
\hypersetup{pdfborder={0 0 0},
            breaklinks=true,%
            bookmarks=true,%
            bookmarksnumbered=true,%
            colorlinks=true,%
            citecolor=red,%
            linkcolor=blue,%
}
\usepackage{graphicx} 
\usepackage{epsfig,color} 
\usepackage{amsmath} 
\usepackage{amssymb,mathrsfs}  
\usepackage{theorem}
\usepackage[usenames,dvipsnames]{xcolor}
\usepackage{algorithm,algorithmic,url}
\newtheorem{definition}{Definition}

\usepackage{eqparbox}

\newenvironment{proof}{\noindent{\bf Proof:}\rm}{\hfill\hfill$\Box$\par\medbreak}
\newtheorem{theorem}{Theorem}[section]

\newtheorem{corollary}[theorem]{Corollary}
\newtheorem{assumption}[theorem]{Assumption}
\newtheorem{lemma}[theorem]{Lemma}
 \newtheorem{remark}[theorem]{Remark}

\newcommand\Item[1][]{%
	\ifx\relax#1\relax  \item \else \item[#1] \fi
	\abovedisplayskip=0pt\abovedisplayshortskip=0pt~\vspace*{-\baselineskip}}

\newcommand{\real}{\mathbb{R}}

\newcommand{\E}{\mathbb{E}}

\newcommand{\Proj}{\operatorname{Proj}}

\newcommand{\argmin}{\operatorname{argmin}}

\newcommand{\subscr}[2]{{#1}_{\textup{#2}}}
\newcommand{\norm}[1]{\|#1\|}
\newcommand{\abs}[1]{|#1|}

\newcommand{\sdb}[1]{{\color{black}{#1}}}
\newcommand{\as}[1]{{\color{black}{#1}}}

\long\def\cut#1{{}}

\title{\sdb{Convergence Analysis of Iterated Best Response \\ for a Trusted Computation Game}\thanks{A preliminary version of this paper appeared at the \emph{2015 American Control Conference} under the title \emph{Trusted Computation with an Adversarial Cloud}~\cite{SDB-AS-CL:14}. Corresponding author: Shaunak~D.~Bopardikar, United Technologies Research Center, Tel. +1-860-610-7858, Email: \texttt{bopardsd@utrc.utc.com}. This work was funded by United Technologies Research Center (UTRC). Alberto Speranzon (email: \texttt{alberto.speranzon@honeywell.com}) was with UTRC at the time this work was carried out. He is now with Honeywell Aerospace. Cedric Langbort (email: \texttt{langbort@illinois.edu}) is with the University of Illinois at Urbana Champaign.}} 

\author{Shaunak D. Bopardikar \qquad Alberto Speranzon \qquad C\'edric Langbort}

\begin{document}
\maketitle

\begin{abstract}  
We introduce a game of trusted computation in which a sensor equipped with limited computing power leverages a central node to evaluate a specified function over a large dataset, collected over time. We assume that the central computer can be under attack and we propose a strategy where the sensor retains a limited amount of the data to counteract the effect of attack. We formulate the problem as a two player game in which the sensor (defender) chooses an optimal fusion strategy using both the non-trusted output from the central computer and locally stored trusted data. The attacker seeks to compromise the computation by influencing the fused value through malicious manipulation of the data stored on the central node. We first characterize all Nash equilibria of this game, which turn out to be dependent on parameters known to both players. Next we adopt an Iterated Best Response (IBR) scheme in which, at each iteration, the central computer reveals its output to the sensor, who then computes its best response based on a linear combination of its private local estimate and the untrusted third-party output. We characterize necessary and sufficient conditions for convergence of the IBR along with numerical results which show that the convergence conditions are relatively tight.
\end{abstract}

\section{Introduction}

The Internet of Things (IoT) \cite{IOT2015} is the next generation internet where many embedded devices (sensors, actuators, controllers, etc.) are interconnected and can exchange data. Although such devices are becoming increasingly more advanced and capable, the amount of data they can process is still a small fraction of what they can collect. As a consequence, IoT devices may need to leverage intermediate but more capable devices that can store and compute over larger data streams.

\smallskip

We study the problem where a sensor (short for IoT device) exchanges data with a larger and more powerful computer to carry out a specific computation on the data the sensor has collects over a certain period of time. Ideally, the sensor would be able to send/stream the data to the computer, which would then compute the function and send back the result. However, we assume here that the data stored in the central computer can be manipulated maliciously by an attacker. In this case, the sensor has three options: 1) to retain a small amount of the data and compute the function of interest on such trustworthy data knowing that the computation will not be as accurate or, 2) take the risk that the attack is only mildly compromising the data on the central computer so that the result of the computation is close to the true value; or, 3) try to exchange partial results iteratively with the central computer and fuse locally trusted computation on small sample with tampered computation on the full dataset. 

\smallskip

This paper formalizes the third scenario, which has the other two as limiting cases. In particular, we model this problem of trusted computation as a game between the sensor/central computer and the attacker. We design and analyze a simple protocol in which each player plays its best response and we study  its convergence properties. Additionally, this paper makes a ``worst case'' assumption, namely that the attacker knows exactly the fusion strategy adopted by the sensor.

\smallskip

The approach we consider in this paper is related to an emerging field called \emph{adversarial machine learning}, wherein two parties, a learner and an attacker, are involved, see~\cite{LH-ADJ-BN-BIPB-JDT-11,BB-BN-PL-12,CS-TS-MB-TS-13}. The learner uses data to train, for example, a classifier or a regressor, and the attacker is modifying the data so that the learner ends up training the algorithm incorrectly. In this context, the problem is posed as a Bayesian game~\cite{CS-TS-MB-TS-13}, where the learner minimizes the effect of the attack on the learning algorithm, whereas the attacker maximizes the deviation of such learning algorithm from the correct result and towards a strategically chosen outcome, under the assumption that only a subset of the data can be modified. In~\cite{CS-TS-MB-TS-13}, for example, the e-mail spam problem is considered, where the learner is set to train a classifier to discriminate between spam and non-spam, while the attacker tries to maximize the chances that a spam is classified as non-spam. 

\smallskip

The approach considered in this paper also relates to the procedure of \emph{fictitious play (FP)}. In this procedure, each player tries to learn the probability distribution from which the opponent is drawing its actions~\cite{GWB:51,JR:51}. A recent body of work in the control literature analyzes convergence of fictitious play for several scenarios~\cite{JSS-GA:04,JSS-GA:05,JRM-GA-JSS:09}. In particular,~\cite{JSS-GA:04} presents unified energy-based convergence proofs that work for several special classes of games under FP. In~\cite{JSS-GA:05}, convergence to Nash equilibria is analyzed under the assumption that each player can access the derivatives of the update mechanisms, leading to dynamic FP. A variant of FP, known as Joint Strategy FP, is proposed and the convergence analyzed for several classes of games, especially in high-dimensional spaces, see~\cite{JRM-GA-JSS:09}. More recently, Gaussian cheap talk games, such as~\cite{FF-AT-CL-14}, have been considered. In this context, a sender (adversary) sends corrupted information to a receiver (sensor) under the assumption that the adversary has full knowledge of the receiver's private information.

\subsection{Main Contributions}

The contributions of this paper are four-fold. First, we formulate a new problem on trusted computation within a game-theoretic framework and adopt an Iterated Best Response (IBR)~\cite{DMR-MPW-04,FD-98} algorithm to compute final strategies for the sensor and the attacker. More specifically, we consider a protocol such that at each iteration, the attacker reveals its output to the sensor that then computes its best response as a linear combination of its private local estimate and of the untrusted output. The attacker can then, based on the announced policy of the sensor, decide its best response. There is a clear mismatch in the information pattern between attacker and sensor and, in particular, the fact that the attacker cannot access the realization of the private local estimate of the sensor distinguishes this work from the information pattern considered in cheap talk games~\cite{FF-AT-CL-14}. 

\smallskip

Second, we characterize conditions on the existence of equilibria of the game. These conditions and the equilibria themselves turn out to be functions of all of the problem parameters, viz., the private information belonging to both players. Therefore, to obtain results from a single player's perspective, a third contribution of this paper is to define two notions of convergence for the IBR algorithm, depending upon whether the algorithm converges for some initial value picked by the attacker (\emph{weak convergence}), or for every initial value (\emph{strong convergence}). We derive necessary conditions for weak convergence and sufficient conditions for strong convergence. If the algorithm converges, then it also tells the sensor how to optimally fuse its private estimate with the output. We identify regimes in which some sufficient conditions are also necessary. Numerical simulations indicate that the conditions are relatively tight. 

\smallskip

Fourth and finally, the analysis in this paper allows for a certain level of \emph{mismatch} in the distributions used by the players in computing their respective cost functions. This  generalizes the analysis presented in the preliminary conference version~\cite{SDB-AS-CL:14}, which assumed that the attacker knows precisely the mean of the distribution used by the sensor to compute its private estimate. Additionally, in the special scenario considered in~\cite{SDB-AS-CL:14}, the analysis in this paper recovers the results from~\cite{SDB-AS-CL:14}, and in one of the cases, also improves the result.

\smallskip

Given that the proposed framework requires an iterative process between sensor and the central computer, the algorithm presented in this paper could be suitable for computation algorithms that are iterative in nature so that partial results can be exchanged between sensor and the central computer. Examples are eigenvalue/eigenvector computation, matrix factorization, iterative optimization methods, etc. The connection of this work to control theory lies in the fact that convergence analysis of the IBR essentially leads to a closed loop dynamical system. This aspect is similar in flavor to the set-ups analyzed in~\cite{JSS-GA:04,JSS-GA:05,JRM-GA-JSS:09}. Using geometric relationships to bound the evolution, we determine necessary and sufficient conditions on the parameters involved which will lead to stability from any/some initial conditions.

\subsection{Paper Organization}

The paper is organized as follows. The problem formulation and the proposed approach is described in Section~\ref{sec:problem}. Conditions for the existence of equilibria together with an insightful geometric interpretation are presented in Section~\ref{sec:mismatch}. Convergence results for the IBR algorithm are derived in Section~\ref{sec:conv_analysis} along with supporting numerical results. A special case of equal means of the random variables used by the two players is discussed in Section~\ref{sec:equalmeans}. Finally, conclusions and directions for future research are discussed in Section~\ref{sec:conclusion}. 

\section{Problem Formulation}\label{sec:problem}

The problem scenario is depicted in Figure~\ref{fig:scenario}. This work assumes that the data and the computation to be carried out are such that it is possible for the sensor to compute an estimate of the true output locally using some random subset of the data, and that the statistics (up to the mean with a finite second moment) about how the actual value is distributed given a value of estimate is known to the sensor. For example, suppose that the data~$d \in \real^{N \times M}$, consisting of~$N$ data points each represented by an $M$-dimensional feature vector, is uploaded through a trusted sensor to the third-party computer. During the upload process, the sensor could retain a randomly sparsified sample~$\hat{d}\in \real^{N \times M}$ of the original data~$d$. The data, once stored on the central computer, can be compromised by the attacker leading to a different value~$\bar{d} \in \real^{N \times M}$. The sensor seeks to compute the true value of the function~$y = g(d)\in \real^k$ of the data~$d$ where $g:\real^{N \times M} \rightarrow \real^k$ is an algorithm of interest. For example, $g(d)\in \real$ could be the maximum singular value of a matrix with~$\hat{d}$ being a sparsified sample of~$d$. The sensor has only an approximate knowledge, $\hat{y} = g(\hat{d}) \in \real^k$, of~$y$, obtained from the sparsified data~$\hat{d}$. The third-party computer computes the value~$\bar{y} = g(\bar{d}) \in \real^k$ based on the corrupted data~$\bar{d}$. This paper does not address the problem of how to construct a sparsified sample $\hat{d}$, but rather makes an implicit assumption that for a given number of non-zero elements in $\hat{d}$ and a corresponding sparsification procedure, the distribution of $y$ given $\hat{y}$ can be characterized a priori. The actual construction of a specific sampling procedure for a computation such as the maximum eigenvalue would be a topic of future investigation and we direct an interested reader to \cite{DA-FM:07} and \cite{PD-AZ:11} for related existing results.

\smallskip

The sensor needs to design a fusion function~$\phi : \real^k\times \real^k \rightarrow \real^k$ that provides~$y_\mathrm{fused}$ so that the effect of the attack is minimized, i.e., the sensor seeks to minimize the squared error between $\subscr{y}{fused}$ and~$y$ in an expected sense. On the other hand, the attacker seeks to drive the fused value to a different value~$y_A$ instead of the true~$y$, and therefore, seeks to minimize the squared error between $\subscr{y}{fused}$ and~$y_A$ in an expected sense. The attacker returns a value $\bar{y} \in \real^k$, which will be chosen in order to minimize the mean squared error between $\subscr{y}{fused}$ and~$y_A$.

\begin{figure}[!t]
  \centering
  \includegraphics[width=0.6\hsize]{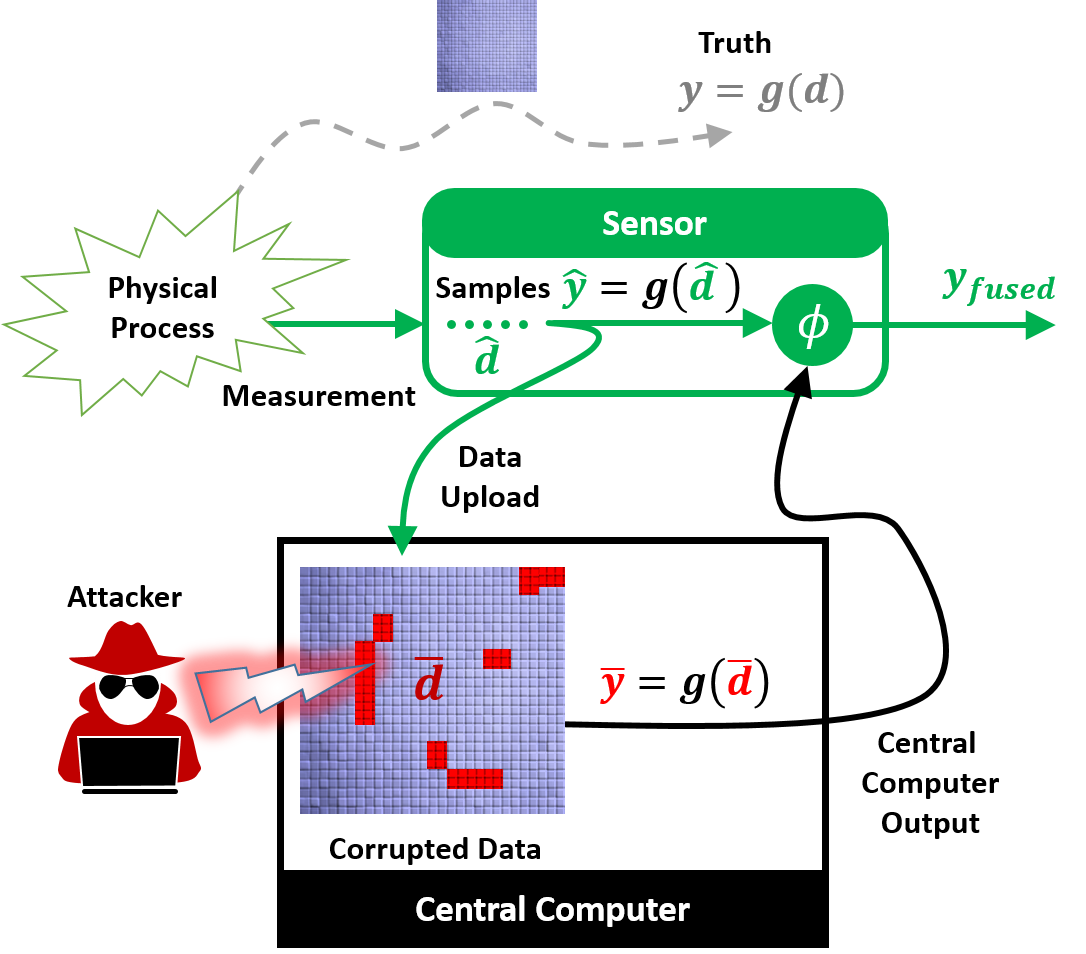}\\
  \caption{Trusted computation over an adversarial third-party computer. The sensor has a large dataset of which it draws a sample~$\hat{d}$ and uploads the entire dataset onto the central computer. The computer returns some value~$\bar{y}$ as the output of the computation, while the sensor computes a value~$\hat{y}$ based on the sample~$\hat{d}$. The goal is to fuse, using a function~$\phi(.)$, the two quantities~$\hat{y}$ and~$\bar{y}$ in an optimal manner. }\label{fig:scenario}
\end{figure}

\smallskip

Given a value of $\hat{y}$, one can view the true value $y$ as a random variable, and we assume that the conditional distribution of $y$ given $\hat{y}$ {\bf ($y|\hat{y}$)} is known to the sensor. From the sensor's perspective, this work only requires it to know the mean of $y$ given $\hat{y}$ and the fact that $y$ given $\hat{y}$ has bounded variance. This paper proposes a new protocol between the sensor and the computer that will enable the sensor to compute an optimal fusion strategy $\subscr{y}{fused} = \phi(\hat{y},\bar{y}) \in \real^k$, in the space of convex combinations of the approximate value~$\hat{y}$ and the computer's output~$\bar{y}$, to minimize the mean square error, i.e.,
\begin{align}\label{eq:sensor_cost}
J_D &:= \E_{y}[\norm{\subscr{y}{fused} - y}^2  \, | \, \hat{y}]\,,
\end{align}
where the norm $\norm{\cdot}$ is the 2-norm.
The goal of the attacker is to pick a $\bar{y}$ to minimize the mean square error 
\begin{align}\label{eq:central_computer_cost}
J_A := \E_{\hat{y}}[ \norm{\subscr{y}{fused} -y_A}^2]\,,
\end{align}
where $y_A$ is a value that the attacker chooses to corrupt the sensor. We will assume that~$y_A$ is a fixed and non-random value known only to the attacker. In other words, the goal of the attacker is to choose a value $\bar{y}$ to be given to the sensor as the \emph{output of the computation} so that when the sensor fuses this value with its private estimate $\hat{y}$, the fused value $\subscr{y}{fused}$ becomes close to a certain value $y_A$ which is selected by the attacker. The value $y_A$ is the value that the attacker wants the sensor to believe is the true output. 
The mismatch in the information structure available to both players is reflected in their respective cost functions in that the random variable relative to which the expectation is computed is the one whose realization is not known to that player. 

\smallskip

In the space of allowed strategies, we choose~$\phi(.)=\phi_\alpha(.)$ to be a convex combination of~$\hat{y}$ and~$\bar{y}$, namely
\begin{equation}\label{eq:linear}
\subscr{y}{fused} := \phi_\alpha(\hat{y},\bar{y}) = \alpha \hat{y} + (1-\alpha) \bar{y}\,.
\end{equation}
In other words, the sensor needs to compute an $\alpha \in [0,1]$, thus deciding the weight to attach to its own value~$\hat{y}$ and the attacker's output~$\bar{y}$. Therefore, both cost functions, $J_D$ and $J_A$ are functions of the action $\alpha \in [0,1]$ of the sensor and the action $\bar{y} \in \real^k$ of the attacker.
To be specific, let~$\mu$ denote the expected value of~$y$ given $\hat{y}$ using a density function chosen by the sensor, modeling the sensor's knowledge about the data, and let~$\zeta$ denote the expected value of~$\hat{y}$ using the density function chosen by the attacker, modeling the knowledge the attacker has about the data. In the space of linear strategies~\eqref{eq:linear}, the cost functions~\eqref{eq:sensor_cost} simplifies to:
\begin{align}\label{eq:JD}
&J_D(\alpha, \bar{y}) = \int_{\real^k} \norm{\alpha \hat{y} + (1-\alpha) \bar{y} - w}^2 f_{ y|\hat{y}}(w) dw \nonumber \\
&= \norm{\alpha (\hat{y} -\bar{y}) + \bar{y}}^2 + \int_{\real^k} \norm{w}^2 f_{ y|\hat{y}}(w) dw\, \nonumber
\\&  - 2 (\alpha (\hat{y} - \bar{y}) + \bar{y})^T \mu 
\end{align}
where~$\mu := \int_{\real^k} w f_{ y|\hat{y}}(w)dw$, and $f_{x|y}(\cdot)$ denotes the probability density function of a random variable~$x$ given~$y$. The cost function~\eqref{eq:central_computer_cost} simplifies to:
\begin{align}\label{eq:JA}
J_A(\alpha, \bar{y}) &= \int_{\real^k} \norm{\alpha w + (1-\alpha) \bar{y} -y_A}^2 f_{\hat{y}}(w) dw \nonumber \\
&= \alpha^2 \int_{\real^k}\norm{w}^2 f_{\hat{y}}(w)dw +\norm{(1-\alpha)\bar{y} -y_A}^2 \nonumber \\ &+ 2\alpha ((1-\alpha)\bar{y} -y_A)^T \zeta, 
\end{align}
where the quantity $\zeta := \int_{\real^k} w f_{\hat{y}}(w)dw$. This is a non zero-sum game for which we will consider the following notion of equilibrium.
\begin{definition}[Nash Equilibrium]\label{def:eqbm}
An admissible pair $(\alpha^*, \bar{y}^*)$ is a \emph{Nash equilibrium} if
\begin{align*}
J_D(\alpha^*, \bar{y}^*) &\leq J_D(\alpha, \bar{y}^*)\,, \qquad \forall \alpha \in [0,1], \text{ and}\\
J_A(\alpha^*, \bar{y}^*) &\leq J_A(\alpha^*, \bar{y})\,, \qquad \forall \bar{y} \in \real^k\,.
\end{align*}
Further, if $\alpha^* \in (0,1)$, then the resulting Nash equilibrium pair is said to be \emph{mixed}. 
\end{definition}

In other words, a pair of strategies is in equilibrium if no other strategy can give a \emph{strictly better} cost against the opponent's strategy, from each player's perspective. Further, when the best response of the sensor is \emph{mixed}, it means that the sensor selects a non-trivial weighted combination of~$\hat{y}$ and~$\bar{y}$. However, it will be clear from the next section that an explicit one-shot computation of the equilibrium strategies requires the players to have full knowledge of all the problem parameters, i.e., $\hat{y}$, $y$, $y_A$, and the expected values~$\mu$ and~$\zeta$ of $y$ given~$\hat{y}$ and~$y$, respectively. Clearly, from one player's point of view, this information is not available. Indeed,~$y_A$ is a private information that the attacker has and is not shared with the sensor and the~$\hat{y}$ is private information of the sensor which is not shared with the attacker. Therefore, we will consider the following iterative scheme in which each player will announce its best response to a strategy announced by the opponent, with the attacker playing first. This protocol is summarized in Algorithm~\ref{algo:IBR}.

\begin{algorithm}[!h]
\caption{Iterated Best Response}
\label{algo:IBR}
\begin{algorithmic}
\STATE \textbf{Assumes:} Attacker plays first, i.e., select a value for $\bar{y}_0$.\\
\STATE $i = 0$
\STATE $\alpha_i = 0$
\STATE \COMMENT{Exit when $\alpha$ reaches steady-state or becomes $1$}
\WHILE{$\alpha_i$ sequence has not converged}
\STATE $ i = i + 1$
\STATE \COMMENT{Client updates its action based on $\bar{y}_{i-1}$}
\STATE $\alpha_i = \argmin_{\alpha \in [0,1]} J_D(\alpha, \bar{y}_{i-1}).$
\STATE \COMMENT{Attacker updates its action based on $\alpha_i$}
\STATE $\bar{y}_i = \argmin_{\bar{y}} J_A(\alpha_i, \bar{y}).$
\STATE \COMMENT{Exit and trust the sensor output only }
\IF {$\alpha_i == 1$}
	\RETURN $\alpha_i$
\ENDIF
\ENDWHILE
\RETURN $\alpha_i$ and $\bar{y}_i$
\end{algorithmic}
\end{algorithm}

\smallskip

To apply Algorithm~\ref{algo:IBR}, we will first need to compute the best responses of each player against an action of the opponent. Setting $\partial J_D/\partial \alpha = 0$ in~\eqref{eq:JD}, we obtain the unconstrained minimizer
\begin{align*}
&(\alpha(\hat{y} - \bar{y}) + \bar{y})^T(\hat{y} - \bar{y}) - (\hat{y} - \bar{y})^T\mu = 0 \Leftrightarrow \subscr{\alpha}{unc}^*(\bar{y}) =\frac{(\mu -\bar{y})^T(\hat{y}-\bar{y})}{\norm{\bar{y} - \hat{y}}^2}\,.
\end{align*}
Due to the constraint $\alpha^* \in [0,1]$, the best response for the sensor is:
\begin{align}
\alpha^*(\bar{y}) &= \begin{cases} 0\,, &\text{if } (\bar{y}-\mu)^T(\hat{y}-\bar{y}) \geq 0\,,  \\
1\,, &\text{if }  (\hat{y}-\mu)^T(\bar{y} - \hat{y})  \geq 0\,, \\
\cfrac{(\bar{y}-\mu)^T(\bar{y}-\hat{y})}{\norm{\bar{y} - \hat{y}}^2}\,, &\text{otherwise.} \end{cases}\label{eq:alpha_star_vec_1}\\
\intertext{A similar calculation for the attacker yields,}
\bar{y}^*(\alpha) &= \begin{cases} \cfrac{y_A-\alpha \zeta}{1-\alpha}, &\text{if } \alpha \neq 1\,, \\
\text{any value}\,,&\text{if } \alpha = 1\,. \end{cases} \label{eq:y_star_vec}
\end{align}
Observe that $\bar{y}^*$ is a linear combination of $y_A$ and $\zeta$. Since we have a total of four parameters $\hat{y}, \zeta, y_A$ and $\mu$, it is convenient considering the plane containing the three points $\hat{y}$,~$\zeta$,~$y_A$ and let $\hat{\mu}$ be the \emph{orthogonal projection} of $\mu$ on to this plane. Therefore, we can write $\mu := \hat{\mu} + \mu^\perp$. Since $\mu{^\perp}^T(\bar{y} - \hat{y}) = 0$, $\alpha^*(\bar{y})$ can be rewritten as
\begin{align}
\label{eq:alpha_star_vec}
\alpha^*(\bar{y}) &= \begin{cases} 0\,, &\text{if } (\bar{y}-\hat{\mu})^T(\hat{y}-\bar{y}) \geq 0\,,  \\
1\,, &\text{if }  (\hat{y}-\hat{\mu})^T(\bar{y} - \hat{y})  \geq 0\,, \\
\cfrac{(\bar{y}-\hat{\mu})^T(\bar{y}-\hat{y})}{\norm{\bar{y} - \hat{y}}^2}\,, &\text{otherwise.} \end{cases} 
\end{align}
Interestingly, we can give a geometric interpretation to~\eqref{eq:alpha_star_vec} and~\eqref{eq:y_star_vec}, as shown in Figure~\ref{fig:geom}. In particular, the figure illustrates the locations of~$\bar{y}$ that lead to different values of~$\alpha^*$, given the values of~$y_A$,~$\hat{y}$ and~$\hat{\mu}$. In the plane defined by $y_A$,~$\hat{y}$ and~$\zeta$, a given subspace containing~$y_A$ can be divided into at most three distinct regions corresponding to the three regimes in~\eqref{eq:alpha_star_vec}. The green regime corresponds to the set of points~$\bar{y}$ for which $\alpha^*(\bar{y}) = 0$ and when it exists, it lies between two blue regimes corresponding to the set of points~$\bar{y}$ for which $\alpha^*(\bar{y}) \in (0,1)$. The red regime corresponds to the set of points~$\bar{y}$ for which $\alpha^*(\bar{y}) = 1$. Note that the boundary points between the green and blue regimes satisfy the property that the lines joining them to~$\hat{\mu}$ and~$\hat{y}$ are orthogonal to each other, as highlighted in Figure~\ref{fig:geom}.

\begin{figure}[!t]
\centering
\includegraphics[width=0.6\hsize]{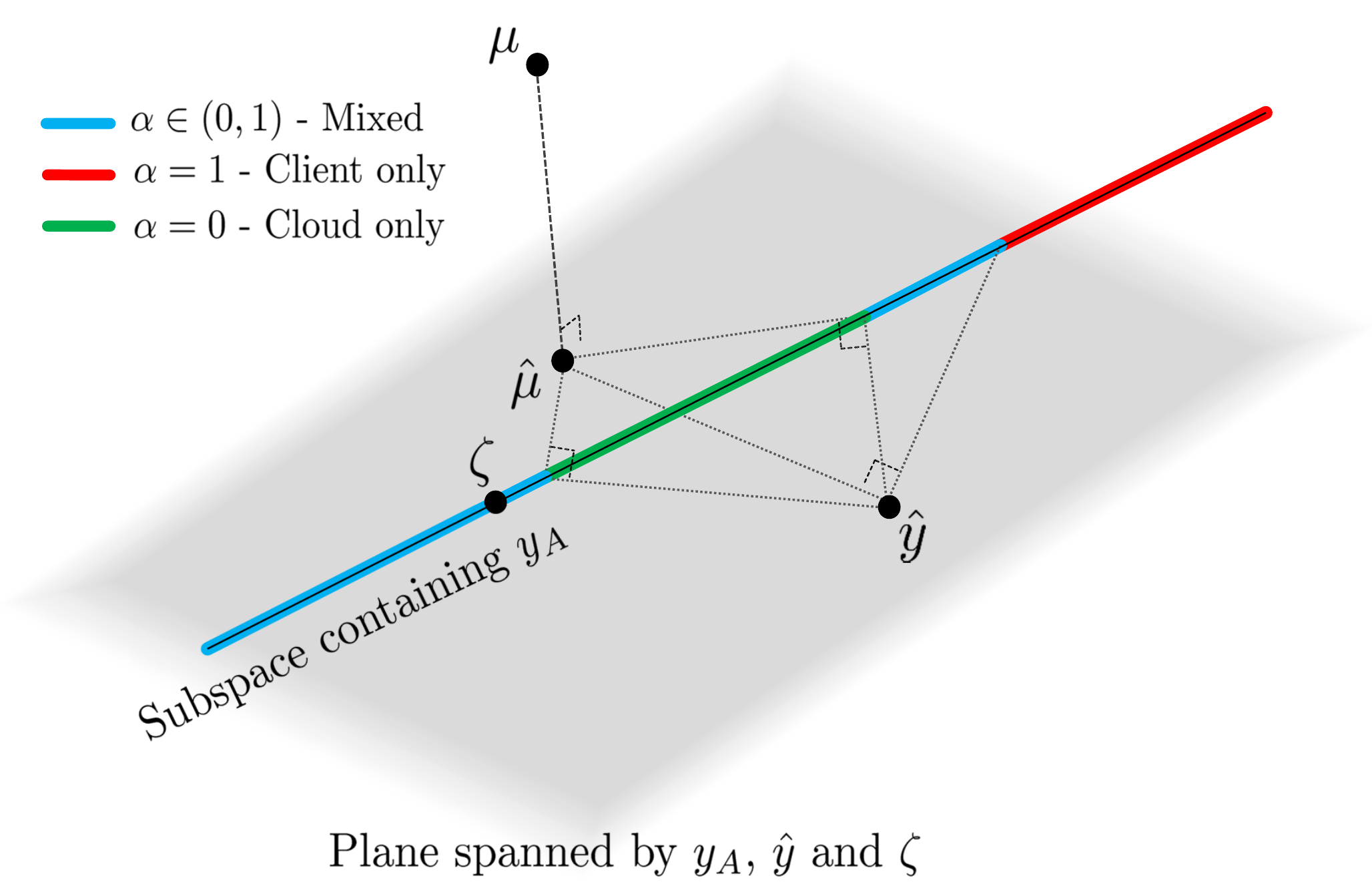}
\caption{The locations for $\bar{y}$ (in the subspace containing~$y_A$) that lead to different values of $\alpha^*$, given the values of~$y_A$,~$\hat{y}$ and~$\mu$. The green locations lead to $\alpha = 0$, the red locations lead to $\alpha = 1$ and the blue locations lead to a mixed $\alpha$.}\label{fig:geom}
\end{figure}

\smallskip

The termination condition within the if loop of Algorithm~\ref{algo:IBR} is due to the fact that when $\alpha = 1$, it implies that the output of the central computer is not to be trusted, and thereafter, the computer is not providing any further useful information than the sensor's private value of~$\hat{y}$. In order to discard trivial initial conditions for which Algorithm~\ref{algo:IBR} would exit at the first iteration, we will restrict our discussion to \emph{non-trivial} initial values of~$\bar{y}_0$ defined next.
\begin{definition}[Non-trivial initial condition]\label{def:non_triv_init}
An initial condition~$\bar{y}_0$ is said to be \emph{non-trivial} if Algorithm~\ref{algo:IBR} does not exit at the first iteration.
\end{definition}
We will consider the following two notions of convergence for Algorithm~\ref{algo:IBR}. 
\begin{definition}[Weak Convergence]\label{def:weak_conv}
Algorithm~\ref{algo:IBR} is said to possess \emph{weak convergence} property if, for \emph{some} non-trivial initial condition $\bar{y}_0$ and for some values of the means~$\mu$ and~$\zeta$, Algorithm~\ref{algo:IBR} outputs $\alpha^* \in [0, 1)$.
\end{definition}
 
\begin{definition}[Strong Convergence]\label{def:strong_conv}
Algorithm~\ref{algo:IBR} is said to possess \emph{strong convergence} property if, for \emph{every} non-trivial choice of $\bar{y}_0$, and for every value of the means~$\mu$ and~$\zeta$, Algorithm~\ref{algo:IBR} outputs $\alpha^* \in [0, 1)$.
\end{definition}

In the sequel, we will see that the weak convergence concept will be useful from the sensor's perspective, whereas strong convergence will be useful from the attacker's perspective. The main contributions in the rest of this paper are to present conditions on the problem parameters, viz. $\hat{y}$ ,$y$, $y_A$, $\mu$ and $\zeta$, under which: 1) equilibrium strategies exist for Algorithm~\ref{algo:IBR} and 2) Algorithm~\ref{algo:IBR} demonstrates weak or strong convergence properties. Whenever Algorithm~\ref{algo:IBR} converges, then the \emph{steady-state} strategies correspond to a Nash equilibrium (cf.~Definition~\ref{def:eqbm}).

\section{Equilibrium Strategies}\label{sec:mismatch}
In this section, we will derive conditions on the parameters $\hat{y}, y_A, \mu$ and $\zeta$ under which equilibria exist for the system of equations~\eqref{eq:alpha_star_vec_1} and~\eqref{eq:y_star_vec}.

Let $\delta := \zeta - \mu$, $z_A:= y_A - \zeta$, and~$\hat{z}:=\hat{y} - \zeta$. The following is the main result on the equilibria of the above system.
\begin{theorem}[Equilibrium]\label{thm:eqbm}
For the system~\eqref{eq:alpha_star_vec_1} and \eqref{eq:y_star_vec}, we have the following:
\begin{enumerate}
\item The pair of strategies $(\alpha^*, \bar{y}^*) = (0, y_A)$ is an equilibrium if and only if
\[
(y_A-\mu)^T(\hat{y}-\mu) \geq \norm{y_A-\mu}^2\,.
\]
\item An equilibrium in mixed strategies exists if and only if 
\[
(\hat{z}^T(\hat{z} + 2z_A - \delta))^2 \geq 4z_A^T(\hat{z} + z_A - \delta)\hat{z}^T\hat{z}\,.
\]
\end{enumerate}
\end{theorem}

\sdb{This result provides conditions on the game parameters under which one of the two types of equilibria considered in this paper would exist. The first type is the one when $\alpha^* = 0$, i.e., the sensor completely trusts the output from the central computer. The condition in (i) essentially describes the set of values for the attacker's intent $y_A$ for which it will be beneficial for the sensor to use the output $\bar{y}$ from the central computer than its own value $\hat{y}$. The second type is the one which will require the sensor to \emph{fuse} its private value $\hat{y}$ with the value $\bar{y}$ in a non-trivial way ($\alpha^* \in (0,1)$).}

\begin{proof}
\begin{enumerate}
\item The claim follows from the best response of the players for the case when $\alpha^* = 0$. In this case, $\bar{y}^* = y_A$. Conversely, if $\bar{y}^* = y_A$, then $\alpha^* = 0$ if and only if
\begin{align*}
&(y_A - \mu)^T(\hat{y} - y_A) \geq 0 \Leftrightarrow (y_A -\mu)^T(\hat{y} - \mu - (y_A - \mu)) \geq 0\,,
\end{align*}
which establishes the first claim.

\item Since we are searching for mixed policies $\alpha^*$, we substitute the expression for $\alpha^*(\bar{y})$ into the fixed point equation for $\bar{y}$ to obtain
\begin{align} \label{eq:zbar_mismatch}
\bar{y} &= \frac{y_A\norm{\hat{y} - \bar{y}}^2}{(\hat{y} - \mu)^T(\hat{y} - \bar{y})} + \frac{(\mu - \bar{y})^T(\hat{y} - \bar{y})}{(\mu-\hat{y})^T(\hat{y} - \bar{y})} \zeta \nonumber \\
\intertext{Subtracting $\zeta$ from both sides, we have}
\bar{y} - \zeta &= \frac{y_A\norm{\hat{y} - \bar{y}}^2}{(\hat{y} - \mu)^T(\hat{y} - \bar{y})} + \Big( \frac{(\mu - \bar{y})^T(\hat{y} - \bar{y})}{(\mu-\hat{y})^T(\hat{y} - \bar{y})} - 1 \Big) \zeta \nonumber \\
\bar{y} - \zeta &= \frac{y_A\norm{\hat{y} - \bar{y}}^2}{(\hat{y} - \mu)^T(\hat{y} - \bar{y})} - \frac{\norm{\hat{y} - \bar{y}}^2}{(\mu-\hat{y})^T(\hat{y} - \bar{y})} \zeta. \nonumber \\
\intertext{By denoting $\bar{z} := \bar{y} - \zeta$, we obtain}
\bar{z} &= \frac{\norm{\hat{z} - \bar{z}}^2}{(\hat{z} - \delta)^T(\hat{z}-\bar{z})} z_A \\
\Rightarrow \bar{z}^* &= rz_A\,, \nonumber
\end{align}
where~$r$ is a scalar that must satisfy
\[
r = \frac{\norm{\hat{z} - rz_A}^2}{(\hat{z} - rz_A)^T(\hat{z}- \delta)}\,.
\]
On simplifying, we obtain the following quadratic equation in $r$:
\begin{align*}
&z_A^T(\hat{z} + z_A - \delta)r^2 - \hat{z}^T(\hat{z} + 2z_A - \delta)r + \hat{z}^T\hat{z} = 0\,.
\end{align*}
The condition now follows from the existence of real roots to the above quadratic equation.
\end{enumerate}
\end{proof}

Clearly, the computation of equilibria requires complete information of the problem parameters. Therefore, in the next sub-section, we will characterize conditions on the parameters under which Algorithm~\ref{algo:IBR} will converge. 

\section{Convergence Analysis}\label{sec:conv_analysis}

In this section, we will derive conditions under which Algorithm~\ref{algo:IBR} possesses weak and strong convergence. 

From the point of view of the sensor, the goal is to characterize conditions on the attack parameter $y_A$ for which Algorithm~\ref{algo:IBR} will converge in the weak and in the strong sense. However, there is an additional uncertainty on where the attacker's mean $\zeta$ will lie. We will commence the convergence analysis for a given value of $\zeta$ and then extend it to the case when there is a bound on the amount of mismatch, such as one in the following assumption.

\begin{assumption}[Mismatch parameter]\label{as:mismatch}
The mismatch $\zeta - \mu $ satisfies
\[
\norm{\zeta - \mu} \leq \frac{\epsilon}{1+\epsilon}\norm{\hat{y} - \mu}\,,
\]
for some given value of $\epsilon \in [0, 1]$.
\end{assumption}
This assumption is reasonable to expect because the deviation $\norm{\hat{y} - \mu}$ will be relatively large when only a small subset of the data being sampled. The convergence analysis will require two intermediate results using geometry which we present in the next sub-section.

\subsection{Preliminary Geometric Results}

Given any $y_A$, we will first show that Assumption~\ref{as:mismatch} leads to the following upper bound on $\norm{\zeta - \hat{\mu}}$, where $\hat{\mu}$ is the orthogonal projection of $\mu$ on to the plane containing $\zeta, y_A$ and $\hat{y}$.

\begin{lemma}[Mismatch bound]\label{lem:mismatch}
For given values of $\zeta, y_A$ and $\hat{y}$, under Assumption~\ref{as:mismatch}, we have 
\[
\norm{\zeta - \hat{\mu}} \leq \epsilon\norm{\hat{y} - \hat{\mu}}\,.
\]
\end{lemma}
\begin{proof}
Consider a realization of the parameters $\zeta, \hat{y}, y_A$ and $\mu$ as shown in Figure~\ref{fig:lem_mismatch}. Since $\hat{\mu}$ is the orthogonal projection of $\mu$ on to the plane containing $\zeta, \hat{y}$ and $y_A$,
\begin{equation}\label{eq:muhat}
\norm{\zeta - \hat{\mu}} \leq \norm{\zeta - \mu}\,.
\end{equation}
Applying the triangle inequality, we obtain
\begin{align*}
\norm{\hat{y} - \hat{\mu}} &\geq  \norm{\hat{y} - \mu} - \norm{\hat{\mu} - \mu} \\
&\geq  \norm{\hat{y} - \mu} - \norm{\zeta - \mu} \\
&\geq (1+\frac{1}{\epsilon}) \norm{\zeta - \mu} - \norm{\zeta - \mu} \\
&=\frac{1}{\epsilon}\norm{\zeta - \mu} \geq \frac{1}{\epsilon}\norm{\zeta - \hat{\mu}}\,,
\end{align*}
where the second inequality follows since $\hat{\mu}$ is the orthogonal projection of $\mu$ on to the plane containing $\zeta, y_A$ and $\hat{y}$, the third inequality follows from Assumption~\ref{as:mismatch} and the final inequality follows from~\eqref{eq:muhat}.
\end{proof}
\begin{figure}[!t]
\centering
\includegraphics[width=0.6\hsize]{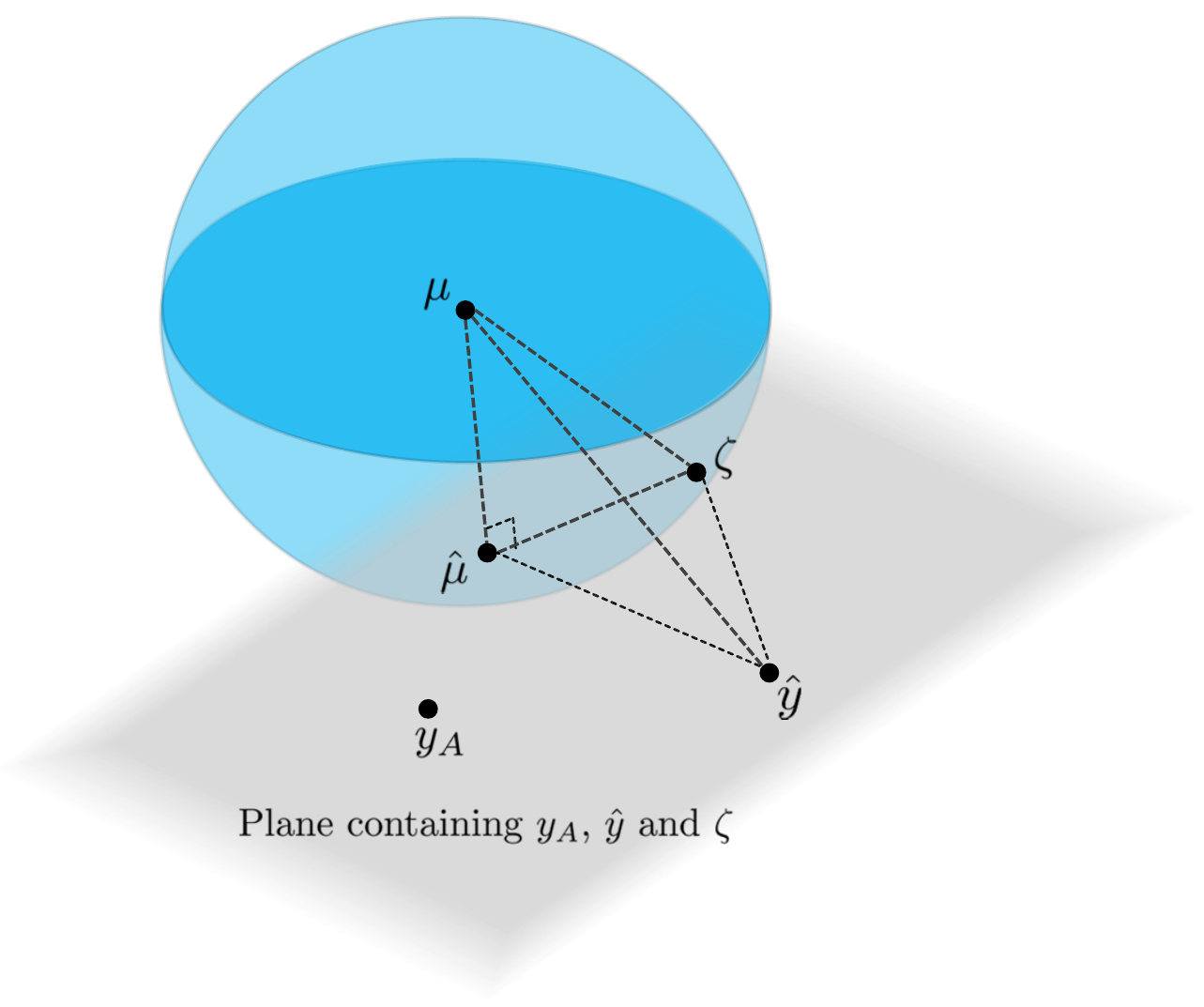}
\caption{Illustrating Lemma~\ref{lem:mismatch}. The attacker's mean, $\zeta$ is assumed to lie inside the sphere centered around $\mu$ and the shaded region illustrates the plane containing the points $\hat{y}, y_A$ and $\zeta$. The point $\hat{\mu}$ is the orthogonal projection of $\mu$ on to this plane.}\label{fig:lem_mismatch}
\end{figure}

We will also require another geometric result which will aid the proof of the necessary condition. We introduce the following notation: given two vectors $x_1, x_2 \in \real^k$, the vector $\Proj(x_1, x_2) \in \real^k$ denotes the orthogonal projection of~$x_1$ onto~$x_2$.

\begin{lemma}\label{lem:ybar}
Suppose that a point $\bar{y}$ satisfies the following:
\begin{enumerate}
\item $\alpha^*(\bar{y}) \in (0,1)$, and
\item $\bar{y}$ lies in the closure of the half plane defined by the line joining $\hat{y}$ and $y_A$ with the side not containing $\hat{\mu}$.
\end{enumerate}
Then, 
\[
\norm{\Proj(\hat{y}-\hat{\mu},y_A - \hat{y})} \geq \norm{\Proj(\hat{y}-\hat{\mu},\bar{y} - \hat{y})}\,.
\]
\end{lemma}
\begin{proof}
We refer the reader to Figure~\ref{fig:lem_ybar} for an illustration. Let the smaller angle between~$y_A-\hat{y}$ and $\bar{y}-\hat{y}$ be~$\alpha$ and the angle between $\hat{\mu}-\hat{y}$ and $y_A-\hat{y}$ be $\beta$. Given the second assumption in the lemma, we have:
\begin{align*}
	\|\Proj(\hat{y}-\hat{\mu},y_A-\hat{y})\| &= \|\hat{y}-\hat{\mu}\| \cos(\alpha+\beta)\\
	\intertext{and}
	\|\Proj(\hat{y}-\hat{\mu},\bar{y}-\hat{y})\| &= \|\hat{y}-\hat{\mu}\| \cos(\beta)\,.	
\end{align*}	
Since  $\alpha^*(\bar{y}) \in (0,1)$, from~\eqref{eq:alpha_star_vec} we conclude that $\alpha+\beta \leq \pi/2$, from which the claim is readily derived.	
\end{proof}
\begin{figure}[!t]
\centering
\includegraphics[width =0.6\hsize]{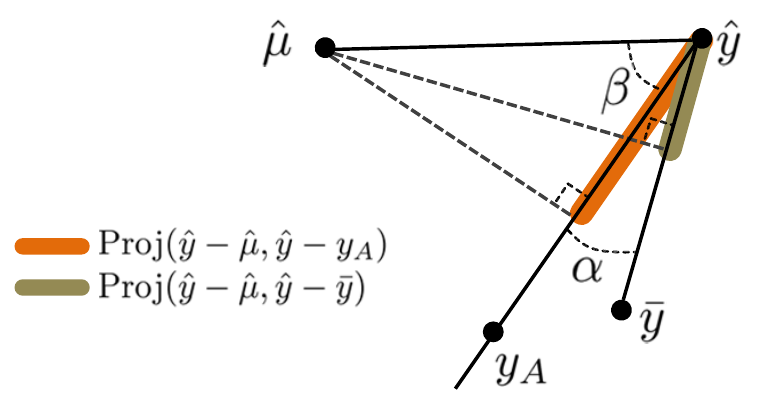}
\caption{Illustrating Lemma~\ref{lem:ybar}. The point $\bar{y}$ lies in the closure of the half plane defined by the line joining $\hat{y}$ and $y_A$ with the side not containing $\hat{\mu}$.}\label{fig:lem_ybar}
\end{figure}

\subsection{Weak Convergence}

We now begin our analysis of a necessary condition assuming that $y_A$ and $\zeta$ are known, which therefore implies that $\hat{\mu}$ is known.

\begin{theorem}[Weak Convergence given $\zeta, y_A$]\label{thm:nec2_mismatch}
For given values of $\zeta, \hat{y}$ and $y_A$, the following hold:
\begin{enumerate}
\item For Algorithm~\ref{algo:IBR} to converge to $(\alpha^* = 0, \bar{y}^* = y_A)$, $y_A$ must satisfy:
\[
(y_A-\hat{\mu})^T(\hat{y}-y_A) \geq 0\,.
\]
\item  Under Assumption~\ref{as:mismatch}, for Algorithm~\ref{algo:IBR} to weakly converge, $y_A$ must satisfy
\[
 (\hat{y}-\hat{\mu})^T(y_A - \hat{y})  < 0\,.
\]
\item Let~$\Psi_{\hat{\mu}}$ be defined as the closure of the half-plane defined by the line joining~$\hat{y}$ and~$y_A$ and which contains the point~$\hat{\mu}$. For Algorithm~\ref{algo:IBR} to yield a converging sequence of mixed $\alpha^* \in (0,1)$, $y_A$ must satisfy either one of the following:\\
\begin{enumerate}
\Item 
\begin{align*}
\norm{y_A-\zeta} \leq \norm{\Proj(\hat{y}-\hat{\mu}, \hat{y}-y_A)}\,,\quad \zeta\in\Psi_{\hat{\mu}}\,,
\end{align*}
\Item
\begin{align*}
\norm{y_A - \zeta} \leq \norm{\hat{y} - \hat{\mu}}\,,\quad \zeta\notin\Psi_{\hat{\mu}}\,.
\end{align*} 
\end{enumerate}
\end{enumerate}
\end{theorem}

\begin{proof}
\begin{enumerate}	
\item Algorithm~\ref{algo:IBR} results into a steady-state value $(\alpha^* = 0, \bar{y}^* = y_A)$ only if (cf.~\eqref{eq:alpha_star_vec_1} and  \eqref{eq:y_star_vec}), 
\[
(y_A - \mu)^T(\hat{y} - y_A) \geq 0 \Leftrightarrow (y_A - \hat{\mu})^T(\hat{y} - y_A) \geq 0\,,
\]
since $\hat{\mu}$ is the orthogonal projection of $\mu$ on to the plane containing $\zeta, y_A$ and $\hat{y}$. This proves the first case.

\item Suppose that Algorithm~\ref{algo:IBR} converges to an $\alpha^*\in [0,1)$ from some arbitrary, non-trivial initial condition $\bar{y}_0$. If~$\bar{y}_0$ is in the green region  ($\alpha_0 = 0$), then after one iteration of Algorithm~\ref{algo:IBR}, we obtain $\bar{y}_1 = y_A$. If~$y_A$ is also in the green region (i.e., corresponding to $\alpha_1 = 0$), then the algorithm converges to an equilibrium with $\alpha^* = 0$ and case i) applies. On the other hand, if~$y_A$ is in the blue region ($\alpha \in (0,1)$), then it follows that $\bar{y}_1 = y_A$, and thus~$\bar{y}_1$ lies in the blue region ($\alpha \in (0,1)$). Therefore, we can assume that without any loss of generality,~$\bar{y}_0$ is in the blue region. From~\eqref{eq:alpha_star_vec}, it needs to hold that
\[
(\hat{y} - \hat{\mu})^T(\hat{y} - \bar{y}_0) > 0\,.
\]
Now, substituting~\eqref{eq:y_star_vec} into~\eqref{eq:alpha_star_vec} and on following the same steps that led to~\eqref{eq:zbar_mismatch}, we have that after one iteration of Algorithm~\ref{algo:IBR}, $\bar{y}_1$ is such that the vector $\bar{y}_0 - \zeta$ is a \emph{positive scalar multiple} of the vector $y_A - \zeta$. More precisely, we have
\[
\bar{y}_1 - \zeta = \frac{\norm{\hat{y} - \bar{y}_0}^2}{(\hat{y} -\hat{ \mu})^T(\hat{y} - \bar{y}_0)}(y_A - \zeta)\,.
\]
Now, observe that the ratio
\begin{equation}\label{eq:init1_mismatch}
\frac{\norm{\hat{y} - \bar{y}_0}^2}{(\hat{y} -\hat{ \mu})^T(\hat{y} - \bar{y}_0)} = \frac{\norm{\hat{y}-\bar{y}_0}}{\norm{\Proj(\hat{y}-\hat{\mu}, \hat{y} - \bar{y}_0)}} > 1\,,
\end{equation}
which implies that for any non-trivial $\bar{y}_0$, $\bar{y}_1 - \zeta$ is a positive scalar multiple \emph{greater than unity} of $y_A - \zeta$. 

Let us assume, for the sake of arriving at a contradiction, that $0 \leq (\hat{y} - \hat{\mu})^T(y_A-\hat{y})$, namely that the condition of case ii) does not hold. Therefore, we can write
\begin{align*}
0 \leq (\hat{y} - \hat{\mu})^T(y_A - \hat{y})  \Rightarrow &0 \leq (\hat{y} - \hat{\mu})^T(y_A - \zeta + \zeta - \hat{\mu} + \hat{\mu}- \hat{y}) \\
\Rightarrow &\norm{\hat{y} - \hat{\mu}}^2 - \norm{\hat{y} - \hat{\mu}}\norm{\zeta - \hat{\mu}}\cos\psi \leq (\hat{y} - \hat{\mu})^T(y_A - \zeta)\,,
\end{align*}
where $\psi$ is the smaller angle between the vectors $\hat{y} - \hat{\mu}$ and $\zeta - \hat{\mu}$. Since Assumption~\ref{as:mismatch} holds, Lemma~\ref{lem:mismatch} holds. Therefore, the left hand side of the above inequality is non-negative, and thus,
\begin{equation}\label{eq:y_1}
0 \leq (\hat{y} - \hat{\mu})^T(y_A - \zeta) \Rightarrow 0 \leq (\hat{y} - \hat{\mu})^T(\bar{y}_1-\zeta),
\end{equation}
since $\bar{y}_1 - \zeta$ is a positive scalar multiple of $y_A - \zeta$. 

Now, from~\eqref{eq:init1_mismatch}, we can write 
\begin{align*}
	0 &\leq  (\hat{y} - \hat{\mu})^T(y_A-\hat{y}) = (\hat{y} - \hat{\mu})^T \left(\frac{1}{r}(\bar{y}_1-\zeta) \vphantom{\frac{1}{r}} + (\zeta-\hat{y})\right) \leq (\hat{y} - \hat{\mu})^T(\bar{y}_1 - \hat{y})\,,
\end{align*}
where the second inequality follows from~\eqref{eq:y_1} and from the fact that $r = \norm{\hat{y} - \bar{y}_0}^2/((\hat{y} -\hat{ \mu})^T(\hat{y} - \bar{y}_0)) \geq 1$. This however implies, from~\eqref{eq:alpha_star_vec}, that $\bar{y}_1$ is in the red region ($\alpha_1 = 1$). But this is not possible as we already showed that $\bar{y}_1$ is in the blue region ($\alpha = (0,1)$). This proves the second case. 


\item From the previous case, without loss of generality, we can assume that $(\bar{y}_0 - \zeta)$ is a positive scalar multiple of $y_A - \zeta$, and that $\bar{y}_0$ is in the blue region (i.e., corresponding to $\alpha \in (0,1)$).  
\begin{figure}[!t]
\centering
\includegraphics[width=0.6\hsize]{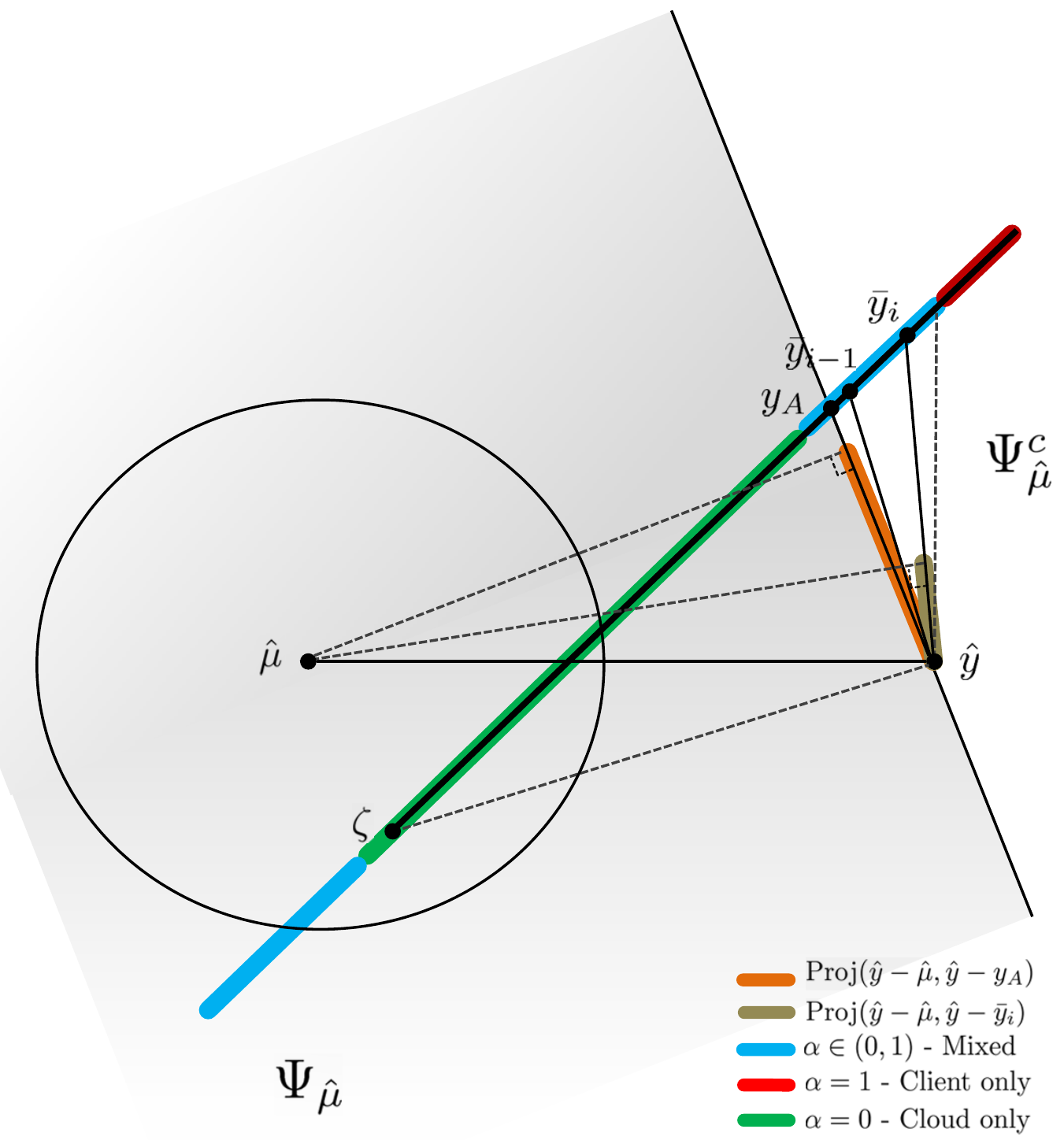}
\caption{Illustrating one possible scenario in the proof of Theorem~\ref{thm:nec2_mismatch}. In this scenario, the points $\zeta$ and $\hat{\mu}$ lie on the same side of the line joining $y_A$ and $\hat{y}$, i.e., the requirement in the third case in the statement of Theorem~\ref{thm:nec2_mismatch} is satisfied, which leads to inequality~\eqref{eq:proj_y_i_mismatch} to hold. In gray, we indicate~$\Psi_{\hat{\mu}}$, the closure of the half-plane defined by the line joining~$\hat{y}$ and~$y_A$ and which contains the point~$\hat{\mu}$.}\label{fig:case2_mismatch}
\label{fig:weak_conv_figure}
\end{figure}

Suppose that for every $i \geq 0$, $\bar{y}_i$ is in the blue region, i.e., every $\alpha_i \in (0,1)$. Then consider the recursion,
\begin{align}\label{eq:yi_mismatch}
\norm{\bar{y}_i - \zeta} = \frac{\norm{\hat{y} - \bar{y}_{i-1}}^2 \norm{y_A - \zeta}}{\abs{(\hat{y} - \hat{\mu})^T(\hat{y} - \bar{y}_{i-1})}} 
\Leftrightarrow \norm{\bar{y}_i - \zeta} = \frac{\norm{\hat{y} - \bar{y}_{i-1}}\norm{y_A - \zeta}}{\norm{\Proj(\hat{y}-\hat{\mu}, \hat{y} - \bar{y}_{i-1})}}\,,
\end{align}
 Note that~\eqref{eq:yi_mismatch} may be viewed as a \emph{discrete-time dynamical system} in $\bar{y}$. Replacing $\bar{y}_0$ by $\bar{y}_{i-1}$ in~\eqref{eq:init1_mismatch}, we have that for every $i \geq 1$, and for any $\bar{y}_0$, such that $\bar{y}_0 - \zeta$ is a positive scalar multiple of $y_A - \zeta$,
\begin{equation}\label{eq:init2_mismatch}
\norm{\bar{y}_{i-1}-\zeta} > \norm{y_A - \zeta}\,.
\end{equation}
We need now to distinguish two cases. 
\item[iii.a] If~$\zeta\in\Psi_{\hat{\mu}}$, then~\eqref{eq:init2_mismatch} implies that for every~$i \geq 1$, the point~$\bar{y}_{i-1}\in \Psi^c_{\hat{\mu}}$, namely the complement of $\Psi_{\hat{\mu}}$. We refer to Figure~\ref{fig:weak_conv_figure} for a geometric interpretation of the proof. Additionally, we have that $\alpha^*(\bar{y}_{i-1}) \in (0,1)$. Thus, applying Lemma~\ref{lem:ybar}, we conclude that for every $i \geq 1$,
\begin{equation}\label{eq:proj_y_i_mismatch}
\norm{\Proj(\hat{y} - \hat{\mu}, \hat{y} - \bar{y}_{i-1})} \leq \norm{\Proj(\hat{y} - \hat{\mu}, \hat{y} - y_A)}\,.
\end{equation}
Further, from applying the triangle inequality to the three points $\bar{y}_{i-1}, \hat{y}$ and $\zeta$, we have
\begin{equation}\label{eq:sinerule_mismatch}
\norm{\hat{y} - \bar{y}_{i-1}} \geq \norm{\bar{y}_{i-1} - \zeta} - \norm{\hat{y} - \zeta}\,.
\end{equation}
Combining this together with \eqref{eq:proj_y_i_mismatch} and \eqref{eq:yi_mismatch} yields
\begin{align*}
\norm{\bar{y}_i - \zeta} &\geq \frac{\norm{y_A - \zeta}(\norm{\bar{y}_{i-1} - \zeta} -\norm{\hat{y} - \zeta})}{\norm{\Proj(\hat{y} - \hat{\mu}, \hat{y} - y_A)}}\,,
\end{align*}
which is a linear system in the quantity $\norm{\bar{y}_i - \zeta}$. This implies that Algorithm~\ref{algo:IBR} will output a converging sequence of mixed $\alpha^*$'s only if 
\[
\norm{y_A-\zeta} \leq \norm{\Proj(\hat{y}-\hat{\mu}, \hat{y}-y_A)}\,.
\]
\item[iii.b] If $\zeta \in \Psi^c_{\hat{\mu}}$, then we can apply the following upper bound
\[
\norm{\Proj(\hat{y} - \hat{\mu}, \hat{y} - \bar{y}_i)} \leq \norm{\hat{y}-\hat{\mu}}\,,
\]
which follows from the fact that the length of the projection of a vector~$x$ onto any another vector can never exceed the length of~$x$ itself. Following the same steps as in the previous case, we conclude that the sequence $\norm{\bar{y}_i - \zeta}$ will converge only if 
\[
\norm{y_A-\zeta} \leq \norm{\hat{y}-\hat{\mu}}\,.
\]
This proves the iii) case.
\end{enumerate}
\vspace*{-0.5cm}
\end{proof}
We numerically verify this result by studying the region of divergence for Algorithm~\ref{algo:IBR} in two dimensions. In the planar case, the value of the mean $\mu = \hat{\mu} = (0,0)$, and the sensor's value $\hat{y} = (0.8,0)$. The attacker's mean, $\zeta = (0.3, -0.2)$. For every value of $y_A$ on a grid in the neighborhood of $\hat{\mu}$, we ran Algorithm~\ref{algo:IBR} for a set of different initial values $\bar{y}_0$. The results are summarized in Figure~\ref{fig:convergence_mismatch}. If the algorithm converged\footnote{The convergence condition in Algorithm~\ref{algo:IBR} was approximated by running the while loop for a sufficiently large number of iterations. In our simulations, we used 100 steps.} to an $\alpha^* \in [0,1)$ for some choice of $\bar{y}_0$, then the corresponding point $y_A$ is \sdb{shown as a (green) dot}. Otherwise, it is \sdb{shown as a (red) cross}. The analytically derived necessary conditions from Theorem~\ref{thm:nec2_mismatch} is shown as a black dashed contour. 

\begin{figure}[!t]
\centering
\includegraphics[width = 0.6\columnwidth]{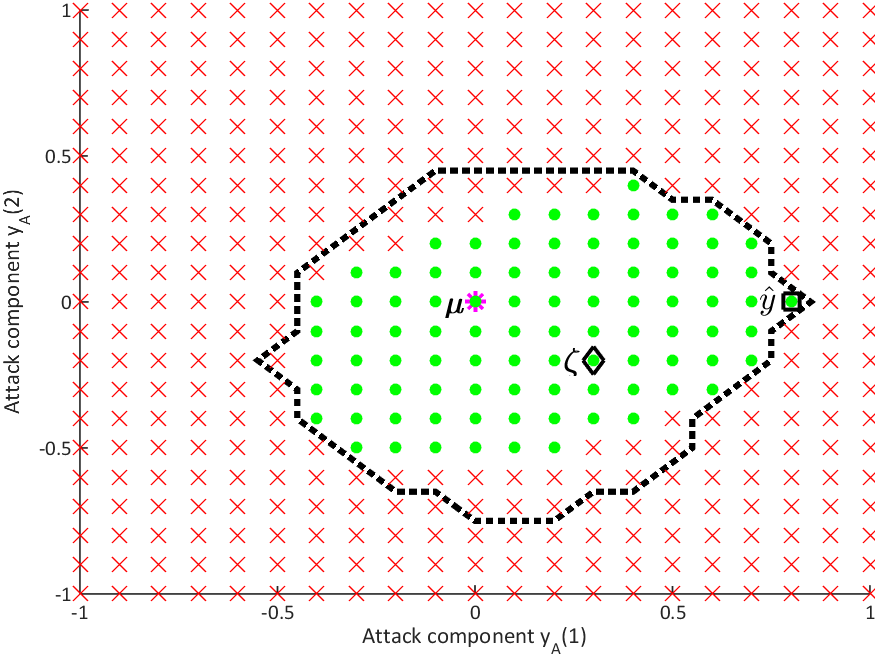}
\caption{Numerically generated plot to study the region of divergence for Algorithm~\ref{algo:IBR} in two dimensions. The value of the mean $\hat{\mu} = (0,0)$ \as{(shown as a (magenta) star)}, the attacker's mean $\zeta = (0.3,-0.2)$ \as{(shown as a (black) diamond)} and the sensor's value $\hat{y} = (0.8,0)$\as{, shown as a ( black) square}. For every value of~\as{$y_A=[y_A(1),y_A(2)]$} on a grid in the neighborhood of~$\mu$, we run Algorithm~\ref{algo:IBR} for a set of different initial values~$\bar{y}_0$. If the algorithm converges to an $\alpha^* \in [0,1)$ for some choice of~$\bar{y}_0$, then the corresponding point~$y_A$ is \as{shown as a (green) dot}. Otherwise, it is \as{shown as a (red) cross}. The analytically derived necessary conditions from Theorem~\ref{thm:nec2_mismatch} is shown as a black dashed contour. }\label{fig:convergence_mismatch}
\end{figure}

We now revisit the fact that the parameter $\zeta$ is not known to the sensor. Suppose that $\zeta \in \mathcal{C}$, where~$\mathcal{C}$ is a set known to the sensor and which satisfies Assumption~\ref{as:mismatch}. Define the set $N_\zeta$ as the set of all points $x\in \real^k$ which satisfy the necessary conditions from Theorem~\ref{thm:nec2_mismatch}. Then, the following result holds.
\begin{corollary}[Necessary Condition]\label{cor:nec}
Suppose that~$\mathcal{C}$ is a set of all $\zeta \in \real^k$ that satisfies Assumption~\ref{as:mismatch}. Then, a necessary condition for weak convergence of Algorithm~\ref{algo:IBR} is that
$$
 y_A \in \bigcup_{\zeta \in \mathcal{C}} N_{\zeta}\,.
$$
\end{corollary}
\begin{proof}
Suppose that Algorithm~\ref{algo:IBR} weakly converges, i.e., for some $\bar{\zeta} \in \mathcal{C}$ and some non-trivial initial condition $\bar{y}_0$, $\alpha^* \in [0,1)$. Then, $y_A$ must satisfy the conditions in Theorem~\ref{thm:nec2_mismatch}. Therefore, $y_A \in N_{\bar{\zeta}}$, and hence, $y_A \in \bigcup_{\zeta \in \mathcal{C}} N_{\zeta}$.
\end{proof}

In Figure~\ref{fig:necfig} we plot the analytic condition from Corollary~\ref{cor:nec} for different sets~$\mathcal{C}$ in $\real^2$. Since this example is planar, $\hat{\mu} = \mu$. In particular,  let the value of the mean $\mu = (0,0)$, the sensor's value $\hat{y} = (1,0)$, and the set~$\mathcal{C}$ be different circles of increasing radii around $\mu$. Figure~\ref{fig:necfig} shows how the set $\bigcup_{\zeta \in \mathcal{C}} N_{\zeta}$ evolves with increasing radius,~$\delta$, of~$\mathcal{C}$. As is expected, the set computed for a higher values of the radius contains the set computed for a lower one. In other words, for a higher level of uncertainty about the attacker's mean $\zeta$, the necessary condition becomes more conservative.

\begin{figure}[!t]
\centering
\includegraphics[width =0.6\columnwidth]{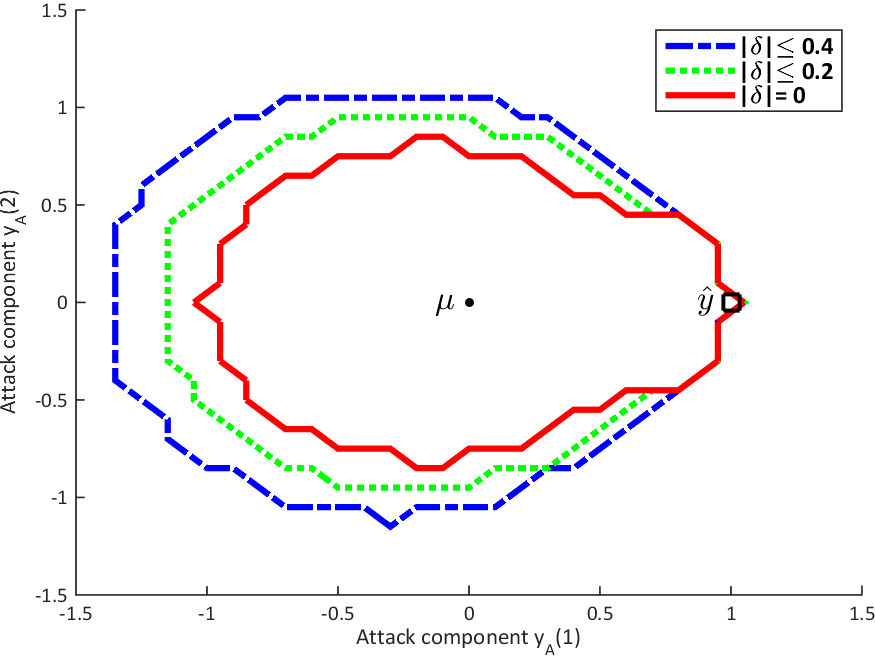}
\caption{Plot of how the analytic necessary condition from Corollary~\ref{cor:nec} evolves for increasing radii of the set~$\mathcal{C}$ which contains $\zeta$. This plot has been numerically generated by sampling $100$ points uniformly randomly out of circles~$\mathcal{C}$ of radii~$\delta$ equal to zero \as{(solid (red) line)}, $0.2$ \as{(dotted (green) line)}, and $0.4$ \as{(dashed (blue) line)}. \as{In this figure, we show $\hat{y}$ as a (black) square and $\mu$ as a (black) dot}.}\label{fig:necfig}
\end{figure}


\subsection{Strong Convergence}

The next result establishes a sufficient condition for strong convergence.

\begin{theorem}[Strong Convergence given $\zeta, y_A$]\label{thm:suff_mismatch}
For given values of $\zeta, \hat{y}$ and $y_A$, Algorithm~\ref{algo:IBR} possesses strong convergence if
\begin{enumerate}
\Item \begin{align}
(y_A - \zeta)^T(\hat{y} - \hat{\mu}) \leq 0\,, \text{ and} \label{eq:suf1_mismatch}
\end{align}
\Item \begin{align}
\norm{y_A - \zeta} \leq \min(&\norm{\Proj(\hat{y}-\hat{\mu}, y_A-\zeta)}, \norm{\Proj(\hat{y}-\hat{\mu}, y_A-\hat{y})})\,.\label{eq:suf2_mismatch}
\end{align}
\end{enumerate}
Additionally, condition~\eqref{eq:suf1_mismatch} is also necessary for strong convergence.
\end{theorem}
\begin{proof}
Following the same steps that lead to~\eqref{eq:init2_mismatch} in the proof of Theorem~\ref{thm:nec2_mismatch}, we can assume without any loss of generality that $\bar{y}_0- \zeta$ is a positive scalar multiple (greater than unity) of  $y_A -\zeta$. Recall~\eqref{eq:yi_mismatch}:
\begin{align}\label{eq:again_yi_mismatch}
\norm{\bar{y}_i - \zeta} &= \frac{\norm{\hat{y} - \bar{y}_{i-1}}}{\norm{\Proj(\hat{y}-\hat{\mu}, \hat{y} - \bar{y}_{i-1})}}\norm{y_A - \zeta}\,.
\end{align}
Observe that in the regime given by \eqref{eq:suf1_mismatch}, the angle between the vectors $y_A - \zeta$ and $\hat{y} - \hat{\mu}$ lies in the interval $[\pi/2, \pi]$. In this regime, for every $i \geq 1$, one out of the following two possibilities occurs.

\begin{enumerate}
\Item \begin{equation}\label{eq:proj_ub_mismatch}
\norm{\Proj(\hat{y} - \hat{\mu}, \hat{y} - \bar{y}_{i-1})} \geq \norm{\Proj(\hat{y} - \hat{\mu}, y_A - \zeta)}\,,
\end{equation}
which holds whenever $\zeta$ is contained in~$\Psi_{\hat{\mu}}$\footnote{Recall that we defined $\Psi_{\hat{\mu}}$ to be the closed half-plane defined by the line joining~$\hat{y}$ and~$y_A$ and contains the point~$\hat{\mu}$.}. This can be seen in Figure~\ref{fig:suff_mismatch}. Equality is achieved when $\bar{y}_{i-1}$ has a magnitude equal to infinity.\\
\begin{figure}[!t]
	\centering
	\includegraphics[width=0.6\hsize]{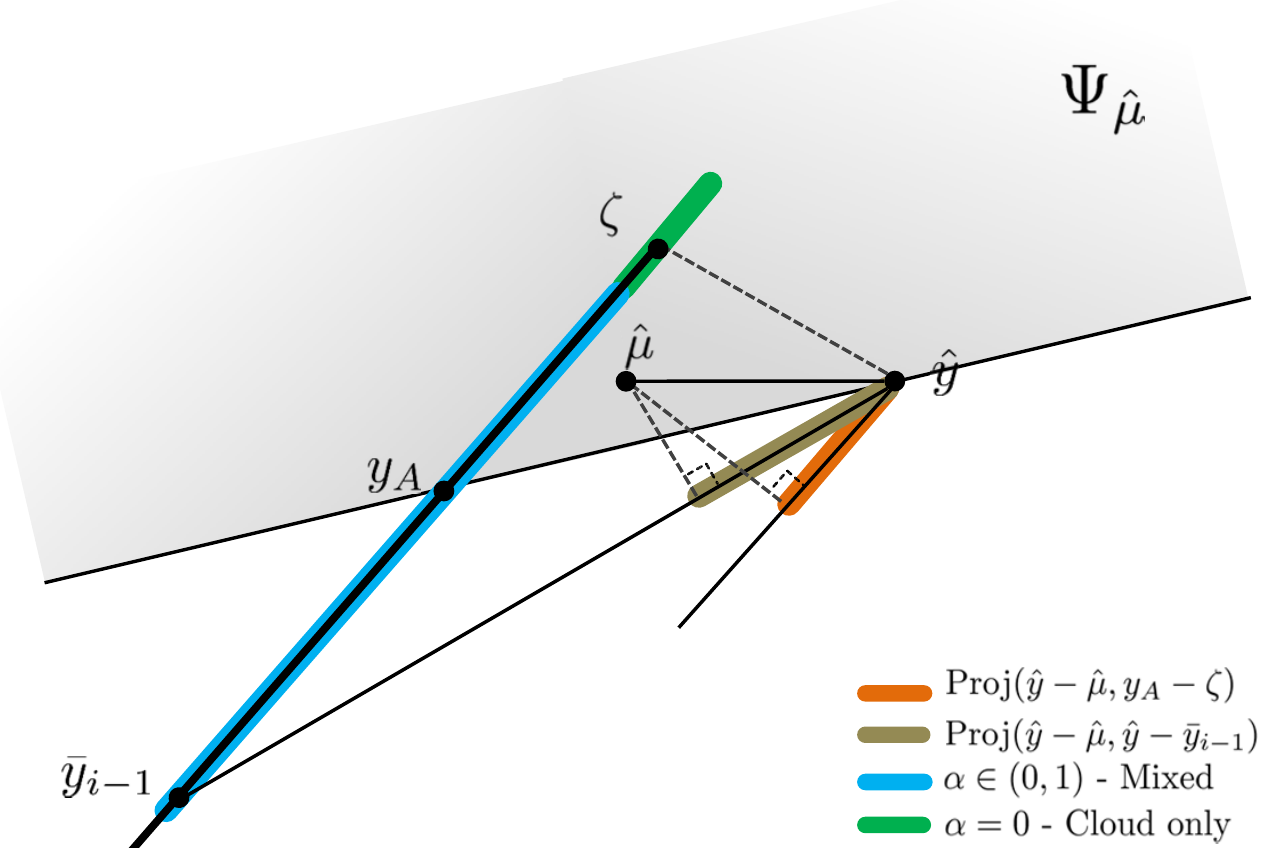}
	\caption{Illustrating first of the two possible configurations: $\zeta \in \Psi_{\hat{\mu}}$ (the closed half-plane defined by the line joining~$\hat{y}$ and~$y_A$ and contains the point~$\hat{\mu}$) in the proof of Theorem~\ref{thm:suff_mismatch}. }\label{fig:suff_mismatch}
\end{figure}
\Item \begin{equation}\label{eq:proj_ub2_mismatch}
\norm{\Proj(\hat{y} - \hat{\mu}, \hat{y} - \bar{y}_{i-1})} \geq \norm{\Proj(\hat{y} - \hat{\mu}, y_A - \hat{y})}\,,
\end{equation}
which holds whenever~$\zeta$ is not contained in the closure of the half plane defined by the line joining $\hat{y}$ and $y_A$ and which contains the point $\hat{\mu}$. This can be seen in Figure~\ref{fig:suff2_mismatch}. Equality is achieved when $\bar{y}_{i-1}$ has a magnitude equal to infinity. 
\begin{figure}[!t]
	\centering
	\includegraphics[width=0.6\hsize]{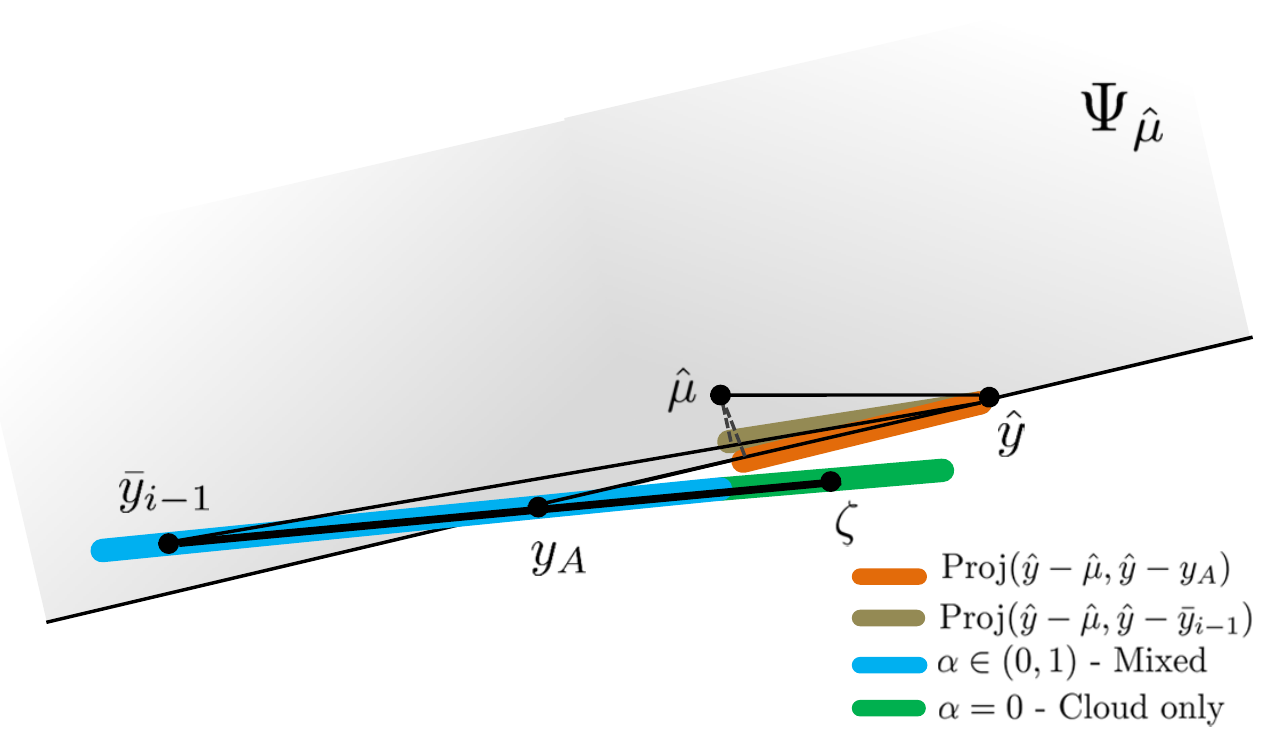}
	\caption{Illustrating the second of the two possible configurations: $\zeta \notin \Psi_{\hat{\mu}}$ (the closed half-plane defined by the line joining~$\hat{y}$ and~$y_A$ and contains the point~$\hat{\mu}$) in the proof of Theorem~\ref{thm:suff_mismatch}. }\label{fig:suff2_mismatch}
\end{figure}
\end{enumerate}
Therefore, we conclude that
\begin{align}\label{eq:lowerbd_mismatch}
&\norm{\Proj(\hat{y} - \hat{\mu}, \hat{y} - \bar{y}_{i-1})} \geq \min(\norm{\Proj(\hat{y}-\hat{\mu}, y_A-\zeta)},\norm{\Proj(\hat{y}-\hat{\mu}, y_A-\hat{y})})\,.
\end{align}
Further, applying the triangle inequality to the set of points~$\zeta$,~$\hat{y}$, and~$\bar{y}_{i-1}$, we obtain
\begin{equation}\label{eq:triangle_suf_mismatch}
\norm{\hat{y} - \bar{y}_{i-1}} \leq \norm{\bar{y}_{i-1} - \zeta} + \norm{\hat{y} - \zeta}\,.
\end{equation}
Combining~\eqref{eq:again_yi_mismatch}, \eqref{eq:lowerbd_mismatch} and \eqref{eq:triangle_suf_mismatch}, we obtain
\begin{align*}
&\norm{\bar{y}_i - \zeta} \leq \frac{\norm{y_A - \zeta} (\norm{\bar{y}_{i-1} - \zeta} + \norm{\hat{y} - \zeta} )}{\min(\norm{\Proj(\hat{y}-\hat{\mu}, y_A-\zeta)}, \norm{\Proj(\hat{y}-\hat{\mu}, y_A-\hat{y})})}\,,
\end{align*}
which converges if condition~\eqref{eq:suf2_mismatch} is satisfied.

We show the necessity of~\eqref{eq:suf1_mismatch} through the following construction. Suppose that condition~\eqref{eq:suf1_mismatch} is violated. Equivalently, suppose that the angle between the vectors $y_A - \zeta$ and $\hat{y} - \hat{\mu}$ is in the interval $[0, \pi/2)$ (we refer the reader to Figure~\ref{fig:case2_mismatch} for an illustration). Now, we can choose a $\bar{y}_0$ to be arbitrarily close to the red region ($\alpha = 1$). The coefficient in front of $\|y_A-\zeta\|$ in ~\eqref{eq:again_yi_mismatch} can be made arbitrarily large. Therefore, in one iteration of the algorithm, $\alpha_1 = 1$. This establishes the necessity of~\eqref{eq:suf1_mismatch}.
\end{proof}

This result is numerically verified in Figure~\ref{fig:sufficiency_mismatch}. In this figure, condition~\eqref{eq:suf1_mismatch} is the vertical line passing through the point $\zeta$ (the black diamond). We can see that there are no green points that lie strictly to the right of this vertical line, thereby numerically verifying the necessity of~\eqref{eq:suf1_mismatch}.

\begin{figure}[!t]
\centering
\includegraphics[width = 0.6\columnwidth]{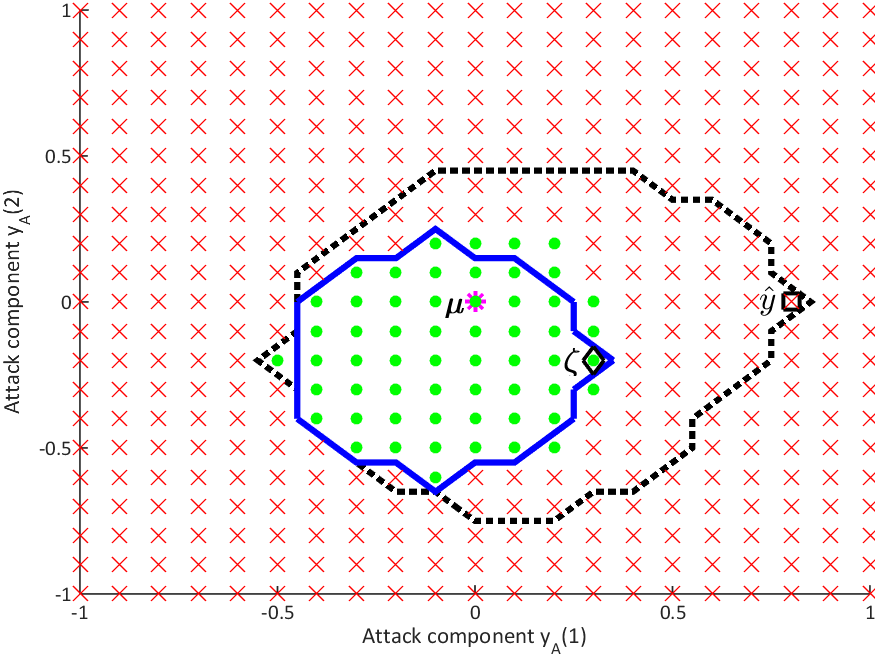}
\caption{Numerically generated plot to study the region of convergence for Algorithm~\ref{algo:IBR} in two dimensions. In the planar case, the value of the mean $\hat{\mu} = \mu = (0,0)$ \as{(shown as a (magenta) star)}, the attacker's mean $\zeta = (0.3, -0.2)$\as{(shown as a (black) diamond)} and the sensor's value $\hat{y} = (0.8,0)$\as{, shown as a (black) square}. For every value of~\as{$y_A=[y_A(1),y_A(2)]$} on a grid in the neighborhood of $\mu$, we run Algorithm~\ref{algo:IBR} for a set of different initial values $\bar{y}_0$. If the algorithm converges to an $\alpha^* \in (0,1)$ for \emph{every} choice of $\bar{y}_0$, then the corresponding point $y_A$ is \as{shown as a (green) dot}. Otherwise, it is \as{shown as a (red) cross}. The analytically derived sufficient condition from Theorem~\ref{thm:suff_mismatch} is shown as a \as{solid (blue) contour}. The analytically derived necessary condition from Theorem~\ref{thm:nec2_mismatch} is shown as a \as{dashed (black) contour}.}\label{fig:sufficiency_mismatch}
\end{figure}

Akin to the necessary condition, we seek a result which does not depend upon the knowledge of $\zeta$. Suppose that $\zeta \in \bar{\mathcal{C}}$, where $\bar{\mathcal{C}}$ is a set known to the sensor. Define the set $S_\zeta$ as the set of all points which satisfy the sufficient conditions from Theorem~\ref{thm:suff_mismatch}. Then, the following result holds.
\begin{corollary}[Sufficient Condition]\label{cor:suf}
A sufficient condition for strong convergence of Algorithm~\ref{algo:IBR} is that
\[
\displaystyle y_A \in \bigcap_{\zeta \in \bar{\mathcal{C}}} S_{\zeta}\,.
\]
\end{corollary}
\begin{proof}
Suppose that  $y_A \in \bigcap_{\zeta \in \bar{\mathcal{C}}} S_{\zeta}$. From Theorem~\ref{thm:suff_mismatch}, this means that for every $\bar{\zeta} \in \bar{\mathcal{C}}$ and for every non-trivial initial condition $\bar{y}_0$, $\alpha^* \in (0,1)$.  Hence, Algorithm~\ref{algo:IBR} strongly converges.
\end{proof}

We now plot the analytic condition from Corollary~\ref{cor:suf} for different sets $\bar{\mathcal{C}}$. In particular,  let the value of the mean $\hat{\mu} = \mu = (0,0)$ (since this is a planar case), the sensor's value $\hat{y} = (0.8,0)$, and the set $\bar{\mathcal{C}}$ be different circles of increasing radii around $\mu$. Then, Figure~\ref{fig:suffig} shows how the set $\bigcap_{\zeta \in \bar{\mathcal{C}}} S_{\zeta}$ evolves with increasing radius of $\bar{\mathcal{C}}$. As is expected, the set computed for a smaller radius contains the set computed for one with higher radius. Since the sufficient condition is of interest to the attacker, we can see that with higher uncertainty in the mismatch between the means, the set of values of $y_A$ that guarantee convergence of Algorithm~\ref{algo:IBR} shrinks in size, making it increasingly difficult for an attacker to always guarantee convergence of Algorithm~\ref{algo:IBR}.

\begin{figure}[!t]
\centering
\includegraphics[width =0.6\columnwidth]{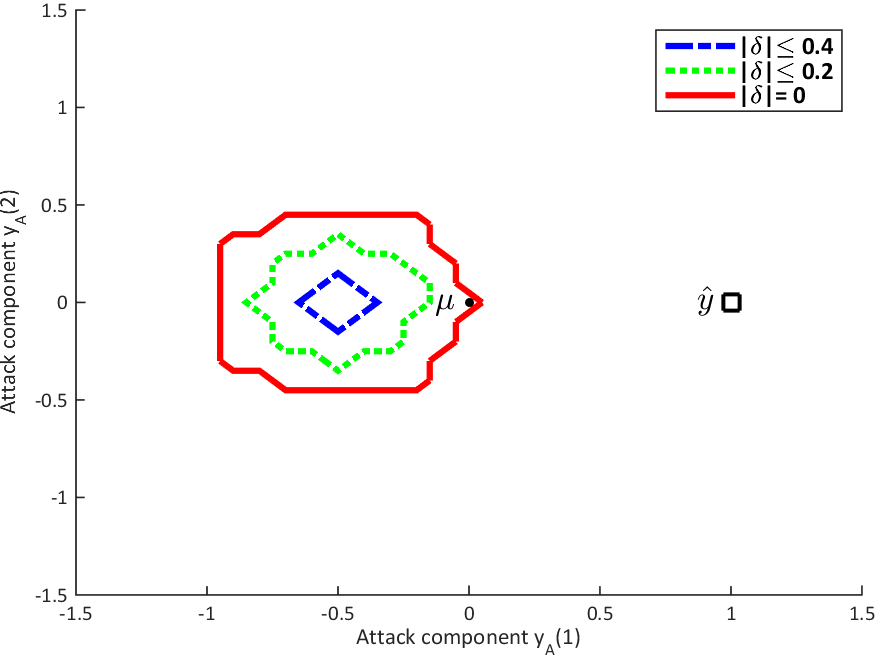}
\caption{Plot of how the analytic necessary condition from Corollary~\ref{cor:suf} evolves for increasing radii of the set~$\bar{\mathcal{C}}$ which contains $\zeta$. This plot has been numerically generated by sampling 100 points uniformly randomly out of circles $\bar{\mathcal{C}}$ of radii equal to zero \as{(solid (red) line)}, $0.2$ \as{(dotted (green) line)}, and $0.4$ \as{(dashed (blue) line)}. \as{In this figure, we show $\hat{y}$ as a (black) square and $\mu$ as a (black) dot}. }\label{fig:suffig}
\end{figure}

\section{A Special Case: Equal Means}\label{sec:equalmeans}
This section addresses the case when the random variables~$y|\hat{y}$ and~$\hat{y}$ have equal means, i.e., $\mu = \zeta$. This scenario may arise in the following situation: the entire multi-sample~$d$ is viewed as multiple independent measurements of a random vector and the trusted sample~$\hat{d}$ drawn by the sensor is not sent to the central computer for computation. This will mean that the random variables $\hat{y}$ and~$y$ are independent with possibly equal expected values (depending on the algorithm $g(.)$ and the sampling distribution for the measurements). However, since~$y$ is computed using many more samples than $\hat{y}$, $y$ is likely to have a lower variance if computed in a non-adversarial setting.

Through a simple substitution of the mismatch parameter $\delta = 0$, it follows that the results from Section~\ref{sec:mismatch} apply directly to this case of equal means. However, the conditions from Theorem~\ref{thm:eqbm} as well as the computation of equilibria still requires complete information of the problem parameters, i.e., $\hat{y}$,~$\mu$ and~$y_A$. So we can study the convergence conditions for Algorithm~\ref{algo:IBR} in this limiting regime, which we now formally present.

\begin{corollary}[Necessary Conditions with $\zeta = \mu$]\label{thm:nec2}
Suppose that $\zeta = \mu$. 
\begin{enumerate}
\item For Algorithm~\ref{algo:IBR} to converge to $\alpha^* = 0$, $\bar{y}^* = y_A$ and~$y_A$ must satisfy
\[
(y_A-\mu)^T(\hat{y}-y_A) \geq 0\,.
\]
\item For Algorithm~\ref{algo:IBR} to weakly converge, $y_A$ must satisfy
\[
  (\hat{y}-\mu)^T(y_A - \hat{y})  < 0\,.
\]
\item For Algorithm~\ref{algo:IBR} to yield a converging sequence of mixed $\alpha^* \in (0,1)$, $y_A$ must satisfy
\[
\norm{y_A-\mu} \leq \norm{\Proj(\hat{y}-\mu, \hat{y}-y_A)}\,.
\]
\end{enumerate}
\end{corollary}
\begin{proof}
When $\zeta = \mu$, Assumption~\ref{as:mismatch} is trivially satisfied. Also, we have that $\hat{\mu} = \mu$.  Additionally, the requirement in the third case of Theorem~\ref{thm:nec2_mismatch} is also trivially satisfied. Therefore, applying Theorem~\ref{thm:nec2_mismatch}, we obtain the claim.
\end{proof}

This result provides an improvement over the conditions presented in our previous work~\cite{SDB-AS-CL:14}. We numerically verify this claim by studying the region of divergence for Algorithm~\ref{algo:IBR} in two dimensions. In this simulation, the value of the mean $\mu = (0,0)$, and the sensor's value $\hat{y} = (1,0)$. For every value of $y_A$ on a grid in the neighborhood of $\mu$, we ran Algorithm~\ref{algo:IBR} for a set of different initial values $\bar{y}_0$. The results are summarized in Figure~\ref{fig:convergence}. If the algorithm converged to an $\alpha^* \neq 1$ for some choice of $\bar{y}_0$, then the corresponding point $y_A$ is colored green. Otherwise, it is colored red. The analytically derived necessary conditions from Corollary~\ref{thm:nec2} are shown as a black dashed contour. We can see that the result is fairly tight especially in the regions where $y_A$ makes an angle that is either very close to zero or to $\pi$. We should expect this to be the case since one can show that the results are exact in the scalar case \cite{SDB-AS-CL:14}.
\begin{figure}[!t]
\centering
\includegraphics[width=0.6\hsize]{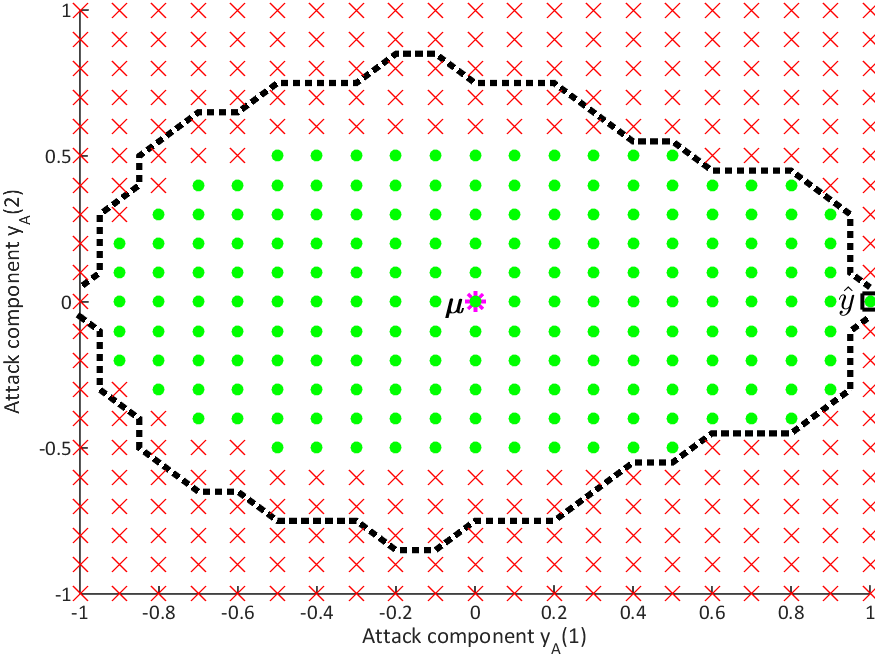}
\caption{Numerically generated plot to study the region of divergence for Algorithm~\ref{algo:IBR} in two dimensions. The value of the mean $\mu = (0,0)$ \as{(shown as a (magenta) star)}, and the sensor's value $\hat{y} = (1,0)$\as{, shown as a (black) square}. For every value of~\as{$y_A=[y_A(1),y_A(2)]$} on a grid in the neighborhood of $\mu$, we run Algorithm~\ref{algo:IBR} for a set of different initial values $\bar{y}_0$. If the algorithm converges to an $\alpha^* \neq 1$ for some choice of $\bar{y}_0$, then the corresponding point $y_A$ is \as{shown as a (green) dot}. Otherwise, it is \as{shown as a (red) cross}. The analytically derived necessary conditions from Theorem~\ref{thm:nec2} are shown as a \as{dashed (black) contour}.}\label{fig:convergence}
\end{figure}

Next, we address the sufficient condition for strong convergence of Algorithm~\ref{algo:IBR} in this case of equal means.
\begin{corollary}[Sufficient Conditions with $\zeta = \mu$]\label{thm:suff}
Suppose that $\zeta = \mu$. Algorithm~\ref{algo:IBR} will converge in the strong sense if the following two conditions are satisfied.
\begin{enumerate}
\Item \begin{align}
(y_A - \mu)^T(\hat{y} - \mu) \leq 0\,,\label{eq:suf1}
\end{align}
\Item \begin{align}
\norm{y_A - \mu} \leq \norm{\Proj(\hat{y}-\mu, y_A-\mu)}\,.\label{eq:suf2}
\end{align}
\end{enumerate}
Additionally, as long as the points $\mu, y_A$ and $\hat{y}$ are \emph{non-collinear}, condition~\eqref{eq:suf1} is also necessary for strong convergence.
\end{corollary}

\begin{proof}
Since $\zeta = \mu$, condition~\eqref{eq:suf1_mismatch} reduces to \eqref{eq:suf1}, while \eqref{eq:suf2_mismatch} reduces to
\begin{align}\label{eq:lowerbd_equalmeans}
\norm{y_A - \mu} \leq \min(&\norm{\Proj(\hat{y}-\mu, y_A-\mu)}, \norm{\Proj(\hat{y}-\mu, y_A-\hat{y})})
\end{align}
We now refer the reader to Figure~\ref{fig:cor_suf} for an illustration of this scenario in which the angle $\gamma$ between the vectors $y_A - \mu$ and $\hat{y} - \mu$ is in the interval $[\pi/2, \pi]$ (i.e., condition~\eqref{eq:suf1} holds). 
\begin{figure}[!h]
\centering
\includegraphics[width=0.6\hsize]{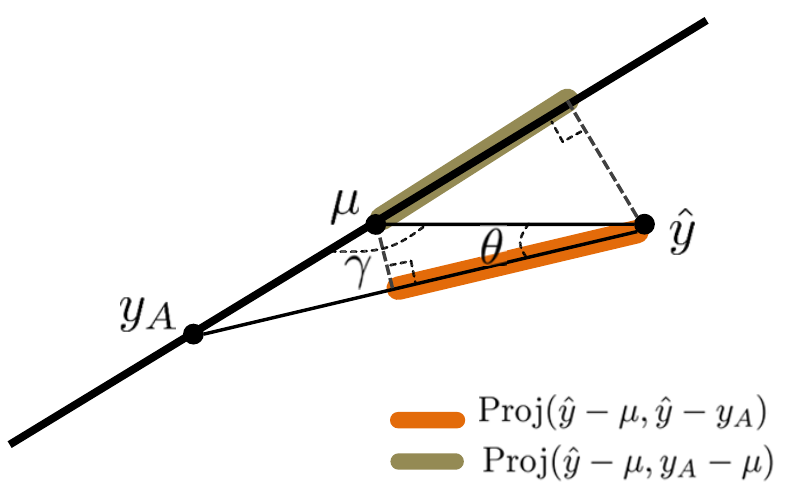}
\caption{Illustration of a geometric property arising in the proof of Corollary~\ref{thm:suff}. In this special case, $\norm{\Proj(\hat{y}-\mu, y_A-\mu)} \leq \norm{\Proj(\hat{y}-\mu, y_A-\hat{y})}$.}\label{fig:cor_suf}
\end{figure}
From the exterior angle property applied to triangle formed by $\hat{y}, \mu$ and $y_A$, we conclude that $\pi - \gamma > \theta$, where $\theta$ is the angle between the vectors $\mu - \hat{y}$ and $y_A - \hat{y}$. Therefore,
\begin{align*}
\norm{\hat{y} - \mu}\cos(\pi-\gamma) \leq \norm{\hat{y} - \mu}\cos{\theta} \Leftrightarrow \norm{\Proj(\hat{y}-\mu, y_A-\mu)} \leq \norm{\Proj(\hat{y}-\mu, y_A-\hat{y})}\,.
\end{align*}
Therefore, \eqref{eq:lowerbd_equalmeans} reduces to~\eqref{eq:suf2}. Applying Theorem~\ref{thm:suff_mismatch}, we obtain the desired claim. The necessity of~\eqref{eq:suf1} follows using the same construction as in the proof of Theorem~\ref{thm:suff_mismatch}.
\end{proof}
\begin{remark}
In case the points $\mu, y_A$ and $\hat{y}$ become \emph{collinear}, the problem reduces to the scalar case described in \cite{SDB-AS-CL:14}.
\end{remark}

This result is numerically verified in Figure~\ref{fig:sufficiency}. We can see that the theoretical conditions have a fairly low gap, especially in the scalar case, i.e., for choices of $y_A$ that lie on the $x$-axis, in which case the result is known to be tight~\cite{SDB-AS-CL:14}. 

\begin{figure}[!t]
\centering
\includegraphics[width=0.6\hsize]{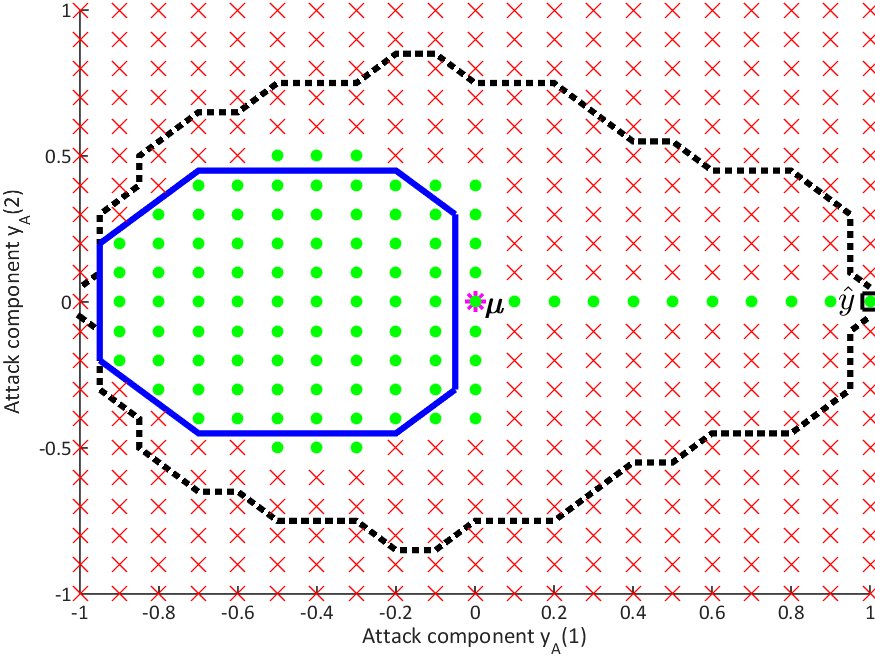}
\caption{Numerically generated plot to study the region of convergence for Algorithm~\ref{algo:IBR} in two dimensions. The value of the mean $\mu = (0,0)$ \as{(shown as a (magenta) star)}, and the sensor's value $\hat{y} = (1,0)$\as{, shown as a (black) square}. For every value of ~\as{$y_A=[y_A(1),y_A(2)]$} on a grid in the neighborhood of $\mu$, we run Algorithm~\ref{algo:IBR} for a set of different initial values $\bar{y}_0$. If the algorithm converges to an $\alpha^* \neq 1$ for \emph{every} choice of $\bar{y}_0$, then the corresponding point $y_A$ is  \as{shown as a (green) dot}. Otherwise, it is \as{shown with a (red) cross}. The analytically derived sufficient conditions from Corollary~\ref{thm:suff} is shown as a \as{solid (blue) contour}. The analytically derived necessary conditions from Corollary~\ref{thm:nec2} are shown as a \as{dashed (black) contour}.}\label{fig:sufficiency}
\end{figure}

\section{Conclusion and Future Directions}\label{sec:conclusion}

This work introduced a novel approach to trusted computation when a central node is leveraged to carry out a complex computation, but is likely to be compromised by an adversary. In our approach, we considered a sensor that may perform approximate but trusted computation on partial data, which is then fused with the output of the central computer in an optimal manner. We proposed a game-theoretic formulation and formalized an iterated best response algorithm. Formal statements were derived that characterize parameter regimes under which the iterative algorithm converges. The derived necessary and sufficient conditions become relatively tight in the case when the distributions of the unknown random variables used by the defender and the attacker to compute their respective cost functions have identical means. Numerical simulations validate our theoretical results.

This work spawns several future directions in this general theme of trusted computation. \sdb{One promising direction is to apply these concepts to a specific computation, such as maximum eigenvalue, to explore tradeoffs between the size of the subset of the original data that needs to be retained by the sensor and the sizes of the regions of convergence.} Further, this paper posed trusted computation as a static problem. It would be of interest to explore dynamic variants of this problem. For example, in this paper, the parameter~$y_A$ is simply chosen by the attacker and is kept fixed in the duration of the game. What if~$y_A$ were to also depend upon the defender's strategy~$\alpha$? It would also be of interest to explore more specific models of attacks in which the attacker can only attack at the computational level or at the data level. \sdb{Another interesting research direction would be to formulate scenarios in which the data owner may also use similar iterative schemes as a \emph{honeypot} \cite{FC:06} to decide whether a particular third-party computation service is trustworthy.}


\end{document}